\documentclass[twocolumn,amsthm]{autart}    

\usepackage{graphicx}          
\usepackage{cite}

\usepackage{amsmath,amssymb,amsfonts}
\usepackage{algorithmic}
\usepackage{dcolumn}	
\usepackage{array}
\usepackage{color}

\usepackage{multirow}
\usepackage{rotating,makecell}
\newcommand{\diff}{\,\mathrm d}
\newtheorem{assumption}{Assumption}
\newtheorem{remark}{Remark}
\newtheorem{theorem}{Theorem}
\newtheorem{lemma}{Lemma}	

\newtheorem{proposition}{Proposition}

\usepackage{caption}
\usepackage{float}
\usepackage{subfig}
\usepackage[title]{appendix}

\begin{document}

\begin{frontmatter}

\title{Safe Delay-Adaptive Control of Strict-Feedback Nonlinear Systems with Application in Vehicle Platooning\thanksref{footnoteinfo}} 

\thanks[footnoteinfo]{This paper was not presented at any IFAC 
meeting.}

\author{Zhenxu Zhao}\ead{23220221151759@stu.xmu.edu.cn}~and~    
\author{Ji Wang}\ead{jiwang@xmu.edu.cn}              

\address{Department of Automation, Xiamen University, Xiamen 361005, China}  

\begin{keyword}                           
Delay-adaptive control, backstepping, safe control, strict-feedback nonlinear systems, vehicle platooning          
\end{keyword}                             										

\begin{abstract}                          
This paper presents a safe delay-adaptive control for a strict-feedback nonlinear ODE  with a delayed actuator, whose dynamic is also a strict-feedback nonlinear ODE and the delay length is unknown. By formulating the delay as a transport PDE, the plant becomes a sandwich configuration consisting of nonlinear ODE-transport PDE-nonlinear ODE, where the transport speed in the PDE is unknown. We propose a predictor-based nonovershooting backstepping transformation to build the nominal safe delay-compensated control, guaranteeing that the output of the distal ODE safely tracks the target trajectory from one side without undershooting. To address the uncertainty in the delay, we incorporate recent delay-adaptive and safe adaptive technologies to build a safe adaptive-delay controller. The adaptive closed-loop system ensures 1) the exact identification of the unknown delay in finite time; 2) the output state stays in the safe region all the time, especially in the original safe region, instead of a subset, after a finite time; 3) all states are bounded, and moreover, they will converge to zero if the target trajectory is identically zero.  In the simulation, the proposed control design is verified in the application of safe vehicle platooning. It regulates the spacing between adjacent vehicles to converge to a small distance and avoids collisions by ensuring they do not breach the safe distance at any time, even in the presence of large unknown delays and at a relatively high speed.
\end{abstract}

\end{frontmatter}

\section{Introduction}
 Vehicle platooning \cite{platoon1}, as an automatic vehicle-following control system that ensures the vehicles in the queue follow each other with a small constant spacing, has gained widespread attention due to its benefits in improving traffic capacity, reducing congestion, and saving fuel. An effective platooning strategy should ensure not only string stability but also safety, i.e., keeping the spacing between adjacent vehicles at a preset safe distance and not breaching this safe distance all the time for the purpose of avoiding collision \cite{platoon2}. Besides, delay compensation is an important issue in vehicle platooning, considering that delays, which have an impact on the system's stability, widely appear in practice, and moreover, its length cannot always be known exactly. The results of safe control design for vehicle platooning under unknown delays are still rare because of the technical challenges in combining delay-adaptive design and CBF-based safe control. 

\subsection{Delay-adaptive control}
Since the emergence of ``Smith predictor"  \cite{bib1}, various control schemes have been developed to reduce the effects of time delays \cite{delay11,bib2}, particularly in nonlinear systems with state delays \cite{delay4,delay3} or input delays \cite{delay5,delay6}. In  \cite{bib4,bib5}, a backstepping-based technique was proposed on the basis of representing the time delay as a transport PDE. Utilizing this technique, the problem of delay compensation in nonlinear systems has been addressed in \cite{delay10, bib5}. This approach has also been extended to compensate for various types of delays, including time-varying delays \cite{varydelay} and state-dependent delays \cite{statedelay,Diagne1,Diagne2}. In addition to the ODE systems, this approach has also been applied in PDE delay compensation, such as in \cite{bib7,bib8,bib9}, where the plant becomes a cascade of PDEs after representing the delay as a transport PDE. More results about nonlinear delay-compensated control are included in \cite{NBbook}.
The aforementioned results consider a known delay. When the exact value of the delay is unknown, a delay-adaptive approach is required to compensate for it. A Lyapunov-based adaptive delay controller was developed for ODE plants \cite{bib12,bib13,bib14,delayadaptive2,delayadaptive3,bib15}. It provides better transient performance than the traditional adaptive methods such as swapping or passive identifiers \cite{bib5}  and has been further developed for PDE systems in \cite{ada2,ada3}. Recently, a delay-adaptive controller for coupled hyperbolic PDE subject to an unknown input delay has been proposed in \cite{delay9}, where a delay estimator is built based on batch least-square identifier (BaLSI) that was introduced in \cite{bib17,bib18} for nonlinear ODEs and extended to PDEs in \cite{bib16,bib19,bib20,bib21}. By this delay estimator, the unknown delay can be exactly estimated in the finite time, which contributes to better transient performance, enabling exponential regulation of the plant states.
\subsection{CBF safe control}
The current delay-adaptive control designs do not consider the safety issue. In many engineering applications,
like autonomous driving, robotics, and UAV, the safety for avoiding collision is vital \cite{safe2,safe4,bib25}.  Control Barrier Functions (CBFs) introduced in \cite{bib23,bib24,bib25}  have been demonstrated as an effective approach to guarantee safety. The CBF-based safe design constrains the focused system state in the safe region by ensuring the nonnegativity of CBFs and then building a safety filter to override the control law. In addition to the above one relative degree CBFs, the high relative degree CBF design was reported in \cite{bib26,nguyen2016exponential} whose root is the nonovershooting control design in \cite{safe1} for a class of strict-feedback nonlinear systems. Utilizing this tool \cite{safe1}, some advanced safe control designs were proposed for the stochastic nonlinear systems  \cite{nonover2}, the Stefan PDE model  \cite{bib30}, or in the prescribed-time safety task \cite{nonover} where safety is only enforced within a preset finite time determined by the user. 
 The aforementioned works do not consider parameter uncertainties in the safe design.
In the case of unknown parameters, guaranteeing safety has attracted the attention of many scholars due to its practical and theoretical significance. The representative work is  \cite{bib29}, which introduces the adaptive Control Barrier Functions (aCBFs), on the basis of the adaptive Control Lyapunov Functions (aCLFs), to ensure adaptive safety. However, it has the conservatism that the plant states are constrained in a subset of the original safe set. The study in \cite{bib32} alleviates this conservatism by leveraging the parameter adaption and data-driven model estimation. Some extended safe-adaptive control results can be found in \cite{bib33,bib37}. Recently, an adaptive-safe control scheme was proposed in \cite{safe5} based on the nonovershooting control design in \cite{safe1} and the BaLSI, which reduces the conservatism of the current safe-adaptive schemes,  constraining the plant states in the original safe set after a finite time, and achieves the exponential regulation of $2\times2$ hyperbolic PDE-ODE cascade, in the presence of the uncertainties in both PDE and ODE subsystems. \par
\subsection{Contributions}
This paper presents the safe delay-adaptive control design for strict-feedback nonlinear ODEs subject to an unknown state delay, ensuring that the output state
safely tracks the target trajectory from one side without undershooting. Main contributions are:\par
1) Unlike the works on safe delay-compensated control in \cite{safe7,safe6,safe8,acc3}  with a known delay, here we achieve safe delay-adaptive control under an unknown delay.\par
2) The existing delay-adaptive control results, such as \cite{bib13,bib14,delay9,wang2021delay,delayadaptivebook}, do not take safety into account, while delay-adaptive regulation and safety guarantee are achieved simultaneously in this paper. \par
3) Compared with the recent result in safe adaptive control design for sandwich ODE-PDE-ODE systems \cite{safe5} where the distal ODE to be safely regulated is linear, the distal ODE is nonlinear in this paper. Besides, we extend the safe stabilization control task in  \cite{safe5} to safe trajectory tracking.\par
4) Different from the previous nonovershooting controller designs \cite{nonover,nonover2,safe1}, the additional nonlinear actuator dynamics and unknown delays are taken into account in our paper.\par
5) To our knowledge, this is the first design of safe delay-adaptive control. The result is new, even for linear systems.
\subsection{Organization}
The problem is formulated in Sec. \ref{sec2}. We present the nominal predictor-based safe control design in Sec. \ref{sec3}. Further, the safe adaptive design for this system, where the delay length is unknown, is proposed in Sec. \ref{secadaptive}. The effectiveness of the proposed design scheme is verified in the application of vehicle platooning with avoiding collisions in Sec. \ref{secexample}. Conclusion and future work are presented in Sec. \ref{con}.
\subsection{Notation} 
\begin{itemize}
    \item  Let $U \in \mathbb{R}^n$ be a set with non-empty interior and let $\Omega \in \mathbb{R}$ be a set. By $\mathbb{C}^0(U; \Omega)$, we denote a class of continuous mappings on $U$, which takes values in $\Omega$. By $\mathbb{C}^k(U; \Omega)$, where $k \ge 1$, we denote the class
of continuous functions on $U$, which have continuous
derivatives of order $k$ on $U$ and take values in $\Omega$.
 	\item The notation $f^{(i)}(t)$ denote $i$  times derivatives of $f$, $u_t^{(i)}(x,t)$, $u_x^{(i)}(x,t)$  denote $i$ times derivatives with respect to $x$  and  with respect to $t$  of  $u(x,t)$ respectively.
  \item Define ${\underline{y}_i(t)}:=[y_1(t),y_2(t),\cdots,y_i(t)]^T$, and $\underline{s}^{(i)}(t):=[s(t),s^{(1)}(t),\cdots,s^{(i)}(t)]^T$
	\item For $n$-vector, the norm $|\cdot|$ denotes the usual Euclidean norm. For square-integrable, measurable functions $u:[0,1]\times\mathbb{R}\rightarrow\mathbb{R}$, the norm $\| u(t)\|:=(\int_{0}^{1}u(x,t)^2\diff x) ^\frac{1}{2}<+\infty $.
	\item  The symbol $e_i$ denotes that n-dimensional unit vector with {\it i} th entry as 1 and other entries are zero, i.e., $e_i:=[\underbrace{0,\cdots,0}_{i-1}, 1,0,\cdots,0]_{1\times n}$.
\end{itemize}
\section{Problem Formulation}\label{sec2}
We consider the following $n+m$ relative order strict-feedback nonlinear system with an unknown state delay $D$, whose position and length are arbitrary:
\begin{align}
	\dot{y_i}(t)&=y_{i+1}(t)+\psi_i({\underline{y}_i}),\, i=1,\cdots,n-1\label{equ1}\\
	\dot{y_n}(t)&=\psi_n({\underline{y}_n})+bx_1(t-D),	\label{equ2}\\
	\dot{x_j}(t)&=x_{j+1}(t)+\varphi_j({\underline{x}_j}), j=1,\cdots,m-1\label{equ3}\\
	\dot{x_m}(t)&=\varphi_m({\underline{x}_m})+U(t),	\label{equ4}
\end{align}
where $Y^T(t)=[y_1,y_2,\cdots,y_n]\in \mathbb{R}^n$ is the state of the ``post-delay" subsystem and  $X^T(t)=[x_1,x_2,\cdots,x_m]\in \mathbb{R}^m $ is the state of the ``pre-delay" subsystem. The unknown delay $D>0$ between the two subsystems satisfies Assumption \ref{assum2}.  The signal  $y_1(t)$ is the output of the overall plant, and the scalar $U(t)$ is the control input to be designed. The nonzero constant $b$ is arbitrary. 
Physically, the $X$-system \eqref{equ3}, \eqref{equ4} driven by the control input describes a nonlinear actuator, i.e., $X$-actuator, whose actuation reaches the nonlinear $Y$-plant subject to an unknown delay. 
\begin{assumption}
	The bound of the unknown parameter is known and arbitrary, i.e., $0<\underline{D}\leq D\leq\overline{D}$, where positive constants $\underline{D}$, $\overline{D}$ are arbitrary and known. \label{assum2}
\end{assumption}
Besides, we make the following assumption about the smoothness of the nonlinear functions $\psi_i$ and  $\varphi_j$, considering a high relative degree plant is dealt with.\par
\begin{assumption}
	The nonlinearity terms $\psi_i({\underline{y}_i})$ in \eqref{equ1}, \eqref{equ2} are $n+m-i$ times continuously differentiable and $\varphi_j({\underline{x}_j})$ in \eqref{equ3}, \eqref{equ4} are $m-j$ times continuously differentiable in all their arguments, and $\psi_i(\textbf{0})=0$, $\varphi_j(\textbf{0})=0$.  \label{assum1}
\end{assumption} 
\textbf{Control objective:} Under the unknown delay $D$,  design a controller $U(t)$ to exponentially regulate the output state $y_1(t)$ to track the target trajectory $s(t)$ and ensure  
\begin{align}
	y_1(t)-s(t)\ge 0,\quad \forall t\ge 0 \label{safe}
\end{align}
i.e., safety defined in this paper, while ensuring that all plant states are bounded. Moreover, when $s(t)\equiv 0$, the exponential convergence to zero of all states in the overall plant is guaranteed. 

Because the plant is $n+m$ relative order, we impose the following assumption regarding the required smoothness of the target trajectory $s(t)$.
\begin{assumption}\label{as:sa}
    The given target trajectory $s(t)$ is $n+m$ times continuously differentiable.
\end{assumption}
For reducing the reading burden in the design process, we denote the distal $Y$ ODE \eqref{equ1}, \eqref{equ2} as
\begin{align}
	\dot Y(t)=f(Y(t),bx_1(t-D)). \label{org11}
\end{align}
Because there is no control actuation on the $Y$-subsystem before $t=D$, we require the following initial condition assumption ensuring the boundedness and safety of  $Y$-subsystem on this no control period $t\in[0,D]$, which is necessary for safe delay-compensated control in nonlinear systems.
\begin{assumption}\label{assum4}
The initial values of plant states satisfy 
  \begin{align}
    P(\theta)<\infty, \label{eq:assum3}
 \end{align}
 and
 \begin{align}
     e_1P(\theta)>s(\theta+D), \label{eq:assum4}
 \end{align}
 for $\theta\in[-D,0]$,
 where
 \begin{align}
	P(\theta)=Y(0)+\int_{-D}^{\theta}f(P(\sigma),bx_1(\sigma))\diff \sigma,\label{pred2}
\end{align}
and where $e_1=[1,0,\cdots,0]_{1\times n}.$
\end{assumption}
Assumption \ref{assum4} is a sufficient and necessary condition of the simulation that the state is still bounded and in the safe region on the uncontrolled period: $t\in[0,D]$, i.e., before the control action kicks in.
Specifically, \eqref{eq:assum3} is a sufficient and necessary condition of the situation that $Y(t)$ does not diverge to infinity within the time interval $[0, D]$,
and \eqref{eq:assum4} is a sufficient and necessary condition of the situation that  $y_1(t)$  remains in the safe region for $t\in[0, D]$. In practice, the solution $P(\theta)$ in \eqref{pred2} can be obtained by the finite difference method, which only depends on the initial conditions $Y(0)$ and initial input history for $Y$-system, i.e., $x_1(t),t\in[-D,0]$. Actually, $P(\theta)$  is the initial condition of the predictor state of $Y$, which will be seen clearly in Sec. \ref{subsec2}.

Besides, we require the following assumption that restricts the actuator state beginning within the region of safe regulation.  
\begin{assumption}
	The initial value of the transport actuator state $x_1(t)$ satisfies $bx_1(0)>\Delta(0)$, where  the value $\Delta(0)$ will be given by \eqref{delta} in Sec. \ref{subsec2}.\label{assum5}
\end{assumption}
Defining 
\begin{align}
	u(x,t)=bx_1(t-D+Dx), \label{uxt}
\end{align}
the delay can be modeled as a transport PDE, and \eqref{equ1}, \eqref{equ2} is rewritten as 
\begin{align}
	\dot Y(t)&=f(Y(t),u(0,t)),\label{org2}\\
	Du_t(x,t)&=u_x(x,t),\label{org3} \\
	u(1,t)&=b{x_1}(t),\label{eq:ux1}
\end{align}
for $x \in [0,1], t\in [0,+\infty)$. Now the overall plant is \eqref{org2}--\eqref{eq:ux1} with \eqref{equ3}, \eqref{equ4}, on the basis of which the control design and stability analysis will be conducted.  
\section{Safe Delay-Compensated Controller} \label{sec3}

\subsection{{Nonundershooting backstepping transformation for the $Y$} part}\label{sec1}

Following \cite{safe1}, we introduce the safe backstepping transformation 
\begin{align}
	&z_i(t)=y_i(t)-h_{i-1}-s^{(i-1)}(t),\label{zi}\\
	&h_0=0,\,h_1=-k_1z_1(t)-\psi_1,\label{eq:h01}\\
	& h_i(\underline{y}_i(t),\underline{s}^{(i-1)}(t))=-k_iz_i(t)-\psi_i({\underline{y}_i})\notag\\
	&+\sum_{k=1}^{i-1}\left[\frac{\partial h_{i-1}}{\partial y_k}\big(y_{k+1}(t)+\psi_k\big)+\frac{\partial h_{i-1}}{\partial s^{(k-1)}}s^{(k)}(t)\right],\notag\\
		&i=1,\cdots,n\label{hi}
\end{align}
where the positive design parameters $k_1,\cdots,k_n$  are to be determined later to ensure safety. We then arrive at
the transformed target $Z$-system given by
\begin{align}
	\dot{z}_i(t)&=-k_iz_i(t)+z_{i+1}(t),\quad i=1,\cdots,n-1 \label{tar4}  \\
	\dot{z}_n(t)&=-k_nz_n(t)-h_n-s^{(n)}(t)+u(0,t). \label{tar5}
\end{align}
The transformed states $Z(t)=[z_1(t),\cdots,z_n(t)]^T$ are indeed high-relative-degree CBFs in \cite{safe5}. Considering the control action begins to regulate the $Y$-system from $t=D$, the CBFs $Z(t)$ need to be kept nonnegative for $t\in[D,\infty]$ in the control design to ensure the safety. \par

\subsection{Predictor-based nonundershooting backstepping transformation for the $X$ part}\label{subsec2}
At the beginning of this subsection, we introduce predictor states $P^T(t)=[P_1(t),\cdots,P_n(t)]$, which denote the $D$ time units ahead predictor of $Y(t)$, i.e.,
\begin{align}
    P(t)=Y(t+D).\label{eq:preP}
\end{align}
These predictor states can be generated by integrating \eqref{org11} from $t$ to $t+D$, as follows: 
\begin{align}
	P(t)=Y(t)+\int_{t-D}^{t}f(P(\sigma),bx_1(\sigma))\diff \sigma,\label{pred1}
\end{align}
with the initial condition given by \eqref{pred2}.
\begin{remark}\label{remark1}
	{\rm In practical implementation, the predictor state $P(t)$ from the above implicit relation \eqref{pred1}, \eqref{pred2} can be solved by the finite difference method, i.e., $P(dk)=Y(dk)+\sum_{i=k-N}^{k-1}f(P(di),bx_1(di))d$, where $d=D/N$ is the discretization step for the integral, $N$ is a free positive integer (a larger $N$ means more accumulations which improve the accuracy of integral but need more computation)}. 
\end{remark}
Based on the predictor \eqref{pred1}, \eqref{pred2}, we define an auxiliary system as
\begin{align}
Dp_t(x,t)&=p_x(x,t),\label{pred4}\\	
 p(1,t)&=P(t),\label{ini1}\\
 p(x,0)&=P(D(x-1)), x\in[0,1]\label{ini2}
\end{align}
whose solution $p^T(x,t)=[p_1(x,t),\cdots,p_n(x,t)]$ is
\begin{equation}
	p(x,t)=P(t-D+Dx),~ x\in[0,1]\times t\in[0,\infty).\label{sol}
\end{equation} 	
Relying on the auxiliary system, applying the following transformation
\begin{align}
		w(x, t)=&u(x, t)-h_n(p(x,t),\underline{s}^{(n-1)}(t+Dx))\notag\\
                &-s^{(n)}(t+Dx)\label{ker}
\end{align}
for $x\in[0,1]$ and  the nonundershooting backstepping transformation
\begin{align}
	&r_j(t)=bx_j(t)-\tau_{j-1}-\Delta^{(j-1)}(t),\label{ri}\\
	&\tau_0=0,\,\tau_1=-c_1r_1(t)-b\varphi_1,\label{tau1}\\
	&\tau_j(\underline{x}_j(t), \underline{\Delta}^{(j-1)}(t))=-c_jr_j(t)-b\varphi_j({\underline{x}_j}) \notag \\
	&+\sum_{k=1}^{j-1}\left[\frac{\partial \tau_{j-1}}{\partial x_k}\big(x_{k+1}(t)+\varphi_k\big)+\frac{\partial \tau_{j-1}}{\partial\Delta^{(k-1)}}\Delta^{(k)}(t)\right],\notag\\
	& j=1,\cdots,m\label{tau}
\end{align}
where  $c_1,\cdots,c_m$ are some positive design parameters to be determined later, and where 
\begin{equation}
		\Delta(t)=	h_n(P(t),\underline{s}^{(n-1)}(t+D))+s^{(n)}(t+D),\label{delta}
\end{equation}
we arrive at the target system
\begin{align}
	\dot{z}_i(t)&=-k_iz_i(t)+z_{i+1}(t),i=1,\cdots,n-1\label{obj1}    \\
	\dot{z}_n(t)&=-k_nz_n(t)+w(0,t),	\label{obj2}\\
	Dw_t(x,t)&=w_x(x,t),\label{obj3}\\
	w(1,t)&=r_1(t),\label{obj4}\\
	\dot{r}_j(t)&=-c_jr_j(t)+r_{j+1}(t),j=1,\cdots,m-1   \label{obj5} \\
	\dot{r}_m(t)&=-c_mr_m(t),		\label{obj6}
\end{align}
with choosing the control law as
\begin{align}
	U(t)&=\frac1b\left(\tau_m(\underline{x}_m(t), \underline{\Delta}^{(m-1)}(t))+\Delta^{(m)}(t)\right)\notag\\
		&:=\mathcal{U}(t,D).	\label{u}
\end{align}
The above nominal control law $U(t)$  is constructed using the plant states $X(t)$, the predictor states $P(t)$, and the derivatives of the trajectory function $s(t)$ up to $(n+m)$ order at the future moment $t+D$, i.e., $\underline{s}^{(n+m)}(t+D)$. We write $D$ as an argument of the function $\mathcal{U}$ \eqref{u} because the predictor $P(t)$ and the trajectory functions $\underline{s}^{(n+m)}(t+D)$ depend on it.
\subsection{Inverse transformations}

In this subsection, we derive the inverse transformations of the transformations in Secs. \ref{sec1}, \ref{subsec2}, i.e., \eqref{zi}--\eqref{hi}, \eqref{ker}, \eqref{ri}--\eqref{tau}, converting the target system back to the original system. First, we show the predictors $\delta^T(x,t)=[\delta_1(x,t),\cdots,\allowbreak\delta_n(x,t)]$ of the transformed states $Z(t)$ in \eqref{obj1}, \eqref{obj2} as follows:
\begin{align}
\delta(x,t)&=Z(t+Dx)\notag\\&=e^{DAx}Z(t)+D\int_{0}^{x}e^{DA(x-y)}Bw(y,t)\diff y \label{inverpred}
\end{align} for $x\in[0,1]\times t\in[0,\infty)$, where
\begin{align}
	A=\begin{pmatrix}
		-k_1&1&0&0&\cdots&0&0\\
		0&-k_2&1&0&\cdots&0&0\\
		&&\ddots&&\vdots&\vdots\\
		0&0&0&0&\cdots&-k_{n-1}&1\\
		0&0&0&0&\cdots&0&-k_n
	\end{pmatrix}_{n\times n}\!\!\!\!\!\!\!\!\!,B=\begin{pmatrix}
		0\\
		0\\
		\vdots\\
		0\\
		1
	\end{pmatrix}_{n\times 1}\!\!\!\!\!\!\!,
\end{align}
which will be used in building the inverse transformations.
\begin{proposition}\label{leminver1}
The inverse transformations of \eqref{zi}--\eqref{hi} are 
\begin{align}
&y_i(t)=z_i(t)+\bar{h}_{i-1}+s^{(i-1)}(t),\label{inveryi}\\
&\bar h_0=0,\,h_1=-k_1z_1(t)-\bar\psi_1,\label{invery0}\\
&\bar{h}_i(\underline{z}_i(t),\underline{s}^{(i-1)}(t))=-k_iz_i(t)-\bar{\psi}_i(\underline{z}_i,\underline{s}^{(i-1)})\notag\\
&+\sum_{k=1}^{i-1}\big(\frac{\partial \bar{h}_{i-1}}{\partial z_k}\big(-k_kz_k+z_{k+1}\big)+\frac{\partial \bar{h}_{i-1}}{\partial s^{(k-1)}}s^{(k)}(t)\big).\notag\\
	&i=1,\cdots,n\label{inverh}
\end{align}
where the functions  $\bar{h}_i(\underline{z}_i,\underline{s}^{(i-1)})$, $\bar{\psi}_i(\underline{z}_i,\underline{s}^{(i-1)})=\psi_i(\underline{y}_i)$ are continuously differentiable in all their arguments and satisfy $\bar{h}_i(0,0)=0$, $\bar{\psi}_i(0,0)=0$.
The inverse of the transformation \eqref{ker} is 
\begin{align}
	u(x, t)=&w(x, t)+\bar{h}_n(\delta(x,t),\underline{s}^{(n-1)}(t+Dx))\notag\\
	&+s^{(n)}(t+Dx),\label{inverker}
\end{align}
where $\delta(x,t)$ is given by \eqref{inverpred}.
	The inverse transformations of \eqref{ri}--\eqref{tau} are 
	\begin{align}
		&x_j(t)=\frac1br_j(t)+\bar{\tau}_{j-1}+\frac1b\bar\Delta^{(j-1)}(t),\label{inver ri}\\
		&\bar\tau_0=0,\,\bar\tau_1=-\frac{c_1}{b}r_1(t)-\bar{\varphi}_1,\label{inver tau1}\\
		&\bar{\tau}_j(\underline{r}_j(t), \underline{\bar{\Delta}}^{(j-1)}(t))=-\frac{c_j}{b}r_j(t)-\bar{\varphi}_j({\underline{r}_j},\underline{\bar{\Delta}}^{(j-1)}) \notag \\
		&+\sum_{k=1}^{j-1}\left[\frac{\partial \bar{\tau}_{j-1}}{\partial r_k}\big(-c_kr_k+r_{k+1}\big)+\frac{\partial \bar{\tau}_{j-1}}{\partial\bar{\Delta}^{(k-1)}}\bar{\Delta}^{(k)}(t)\right],\notag\\
		& j=1,\cdots,m\label{invertau}
	\end{align}
	where
 \begin{equation}
 \bar\Delta(t)=\bar{h}_n(\delta(1,t),\underline{s}^{(n-1)}(t+D))+s^{(n)}(t+D),\label{inverdelta}
 \end{equation} 
 and the functions $\bar{\tau}_j(\underline{r}_j, \underline{\bar{\Delta}}^{(j-1)})$,  $\bar{\varphi}_j({\underline{r}_j}$,$\underline{\bar{\Delta}}^{(j-1)})=\varphi_j({\underline{x}_j})$ are continuously differentiable in all their arguments and satisfy $\bar{\tau}_j(0,0)=0$, $\bar{\varphi}_j(0,0)=0$.
\end{proposition}
\begin{proof}
    The proof is shown in  Appendix \ref{dix1}.
\end{proof}
 According to \eqref{inveryi}--\eqref{inverh}, \eqref{inverker}, \eqref{inver ri}--\eqref{invertau},
we have built the invertible transformation between the $(Y,u,X)$-original system with the predictor depending on $Y,u$ and the $(Z,w,R)$-target system with the predictor $\delta$ depending on $Z,w$, which is illustrated in Fig. \ref{fig trans}.
The invertibility built here will be used to prove the stability of the closed-loop system.
\begin{figure}[t]
	\centering	
	\includegraphics[width=1\linewidth, height=0.1\textheight]{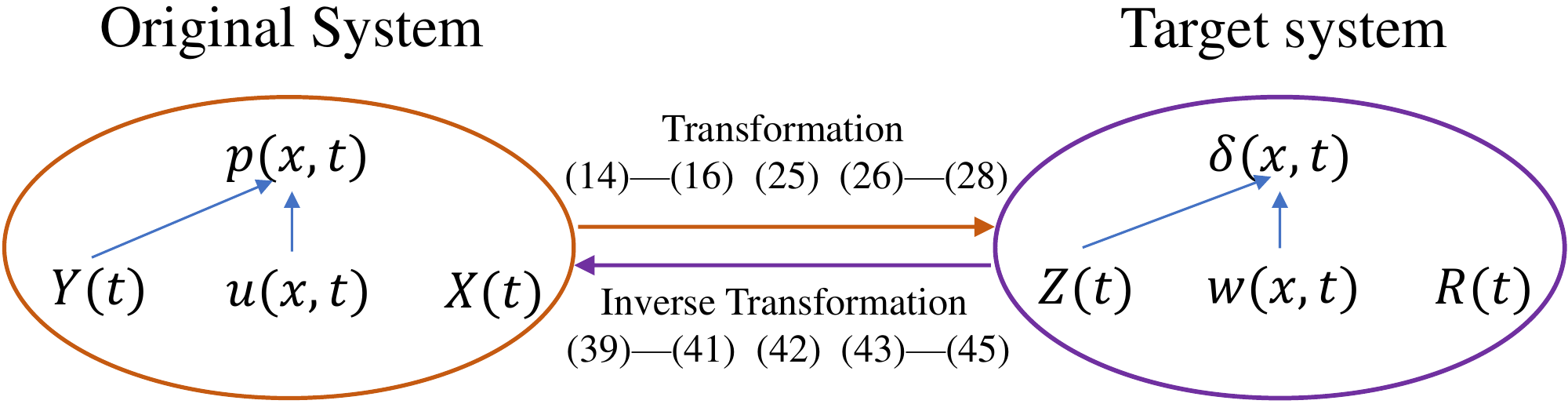}
	\caption{ The transformation between the $Y,u,X$-original  and $Z,w,R$-target systems with the predictors $p,\delta$.}		
	\label{fig trans}
\end{figure}

\subsection{Selection of Design Parameters} \label{selection}
We need to select the design parameters under certain conditions to ensure safety, as shown in this subsection.
We select the parameters $k_1,\cdots,k_n$ to satisfy
\begin{align}
	k_i>\max\{2,\check{k}_i\},i=1,\cdots,n-1,\quad k_n>1\label{ki},
\end{align}
where
\begin{align}
		\check{k}_i&=\frac{1}{P_i(0)-h_{i-1}(\underline{P}_{i-1}(0),\underline{s}^{(i-2)}(D))-s^{(i-1)}(D)}\notag\\
		&\times\Bigl[-P_{i+1}(0)-\psi_i({\underline{P}_i(0)})+s^{(i)}(D)\notag\\
		&+\sum_{k=1}^{i-1}\big(\frac{\partial h_{i-1}}{\partial P_k}\bigl(P_{k+1}(0)+\psi_k\bigl)+\frac{\partial h_{i-1}}{\partial s^{(k-1)}}s^{(k)}(D)\big)\Bigl],\notag\\
		&i=1,\cdots,n-1 \label{k}
\end{align}
and where $P_i(0)=e_iP(0)$ is given in \eqref{pred2}. The purpose of this design parameter selection is to ensure that the CBF states $z_i(t),i=1,\cdots,n$ are positive at the initial time $t=D$ for control action at the distal ODE \eqref{equ1}, \eqref{equ2}. This will be seen clearly in the proof of Lemma \ref{lemmasafe2}.

The design parameters $c_1,\cdots,c_m$ are selected as:
\begin{align}
c_j>\max\{2,\check{c}_j\},j=1,\cdots,m-1,\quad c_m>1,\label{ci}
\end{align}
where
\begin{align}
		\check{c}_j&=\frac{1}{bx_j(0)-\tau_{j-1}(\underline{x}_{j-1}(0), \underline{\Delta}^{(j-2)}(0))-\Delta^{(j-1)}(0)}\notag\\
		&\times \Big[-bx_{j+1}(0)-b\varphi_j({\underline{x}_j}(0))+\Delta^{(j)}(0)\notag\\
		&+\sum_{k=1}^{j-1}\big(\frac{\partial \tau_{j-1}}{\partial x_k}(x_{k+1}(0)+\varphi_k)+\frac{\partial \tau_{j-1}}{\partial \Delta^{(k-1)}}\Delta^{(k)}(0)\big)\Big],\notag \\
		&j=1,\cdots\,m-1. \label{c}
\end{align} 
This selection is to make the values of CBF states $r_j(t),j=1,\cdots,m$  positive at the initial time $t=0$, which will be seen clearly in the proof of  Lemma \ref{lemmasafe1}. Additionally, it also contributes to the exponential regulation of all plant states,  which will be shown in Lyapunov analysis \eqref{lya}. 

Please note that the selection of design parameters $k_1,\cdots, k_n,c_1,\cdots,c_m$ are only determined by the predictor and 
 plant initial values: $P(0)$, $X(0)$, and the denominators in \eqref{k}, \eqref{c} are nonzero that can be seen in the proofs of Lemmas \ref{lemmasafe1}, \ref{lemmasafe2}. 

\subsection{Result with the nominal safe delay-compensated control}\label{result}
\begin{theorem}\label{theo1}
	For initial data $Y(0)\in\mathbb{R}^n$, $X(0)\in\mathbb{R}^m$, the initial input history $x_1(t)\in \mathbb{C}^{m-1}([-D,0])$ satisfying Assumptions \ref{assum4}, \ref{assum5}, and a target trajectory $s(t)$ satisfying Assumption \ref{as:sa}, choosing the design parameters $k_1,\cdots,k_n,c_1,\cdots,c_m$ satisfying \eqref{ki}--\eqref{c}, the closed-loop system \eqref{equ1}--\eqref{equ4} with the nominal controller \eqref{u} and \eqref{pred1}, \eqref{pred2} has the following properties
	\begin{enumerate}
		\item {The output $y_1(t)$ exponentially tracks the target trajectory $s(t)$ in the sense that $|y_1(t)-s(t)|$ exponentially converges to zero, and all plant states, i.e., 
  \begin{align}
\Psi(t)=&|X(t)|+|Y(t)|+\max_{x\in[0,1]}|p(x,t)|\notag\\
&+\left(\int_{t-D}^t x_1(\tau)^2d\tau\right)^{\frac{1}{2}}\label{eq:Psi}
  \end{align} 
  is bounded, and the ultimate bound depends on the target trajectory. If the target trajectory $s(t)\equiv0$, then $\Psi(t)$ is exponentially convergent to zero.}
		\item {The safety is enforced in the sense that $y_1(t)-s(t)\geq0$ holds on $t\geq 0$}. 
	\end{enumerate}
\end{theorem}
Before presenting the proof, we propose the following lemmas regarding the exponential stability and non-negative CBFs in the target system \eqref{obj1}--\eqref{obj6}, which will be used in establishing the theorem.
\begin{lemma}\label{lem:targetstability}
    The exponential stability of the target system \eqref{obj1}--\eqref{obj6} is achieved in the sense that there exist positive constants $\Upsilon_{\Omega}$ and $\sigma_{\Omega}$ such that
    \begin{align}
        \Omega(t)\leq \Upsilon_{\Omega}\Omega(0)e^{-\sigma_{\Omega} t}\label{ome3}
    \end{align} where 
    \begin{equation}
		\Omega(t)=\sum_{i=1}^{n} z_i(t)^2 +\sum_{i=1}^{m} r_i(t)^2+\sum_{i=0}^{m}\Vert w_x^{(i)}(\cdot,t) \Vert^2.\label{ome1}
	\end{equation}
\end{lemma}
\begin{proof}
The proof is shown in Appendix \ref{dixlem1}.
	\end{proof}
 \begin{lemma}\label{lemma:deltarget}
  For the predictor $\delta$ \eqref{inverpred} of the transformed state $Z(t)$ in the target system, the signals $|\delta(x,t)|^2$, $\forall x\in[0,1]$ are exponentially convergence to zero, and also $|\delta_t^{(i)}(1,t)|^2,i=0,\cdots,m$ are exponentially convergent to zero.
 \end{lemma}
 \begin{proof}
 The proof is shown in Appendix \ref{dixlem2}.
\end{proof}
We show that the CBFs are ensured non-negative in the following two lemmas. 
	\begin{lemma}
		The high-relative-degree CBFs $r_j(t), j=1,\cdots,m$ are nonnegative all the time under the selection of design parameters \eqref{ci}, \eqref{c}, i.e., $r_j(t)\ge 0, j=1,\cdots,m$, for  time $t\ge 0$.\label{lemmasafe1}
	\end{lemma}
	\begin{proof}
		Firstly, we claim that all the initial values of $r_j(t)$ are positive. Recalling Assumption \ref{assum5} and \eqref{uxt}, \eqref{ker}, \eqref{delta}, \eqref{obj4}, we have the initial condition $r_1(0)=w(1,0)>0$.
		Setting $t=0$ in \eqref{ri},  we have
		\begin{align}	
				&r_{j+1}(0)=bx_{j+1}(0)-\tau_{j}(\underline{x}_{j}(0),\underline{\Delta}^{(j-1)}(0))-\Delta^{(j)}(0)\notag\\
				&=bx_{j+1}(0)-\Delta^{(j)}(0)+c_jr_j(0)+b\varphi_j(\underline{x}_{j}(0)) \notag\\
				&-\sum_{k=1}^{j-1}\big(\frac{\partial \tau_{j-1}}{\partial x_k}(x_{k+1}(0)+\varphi_k)+\frac{\partial \tau_{k-1}}{\partial \Delta^{(k-1)}}\Delta^{(k)}\big),\label{r0}
		\end{align}
		for $j=1,\cdots,m-1$. For $r_j(0)>0$, using the design parameter selection of $c_j$ in \eqref{ci}, \eqref{c}, we obtain from  \eqref{r0} that $r_{j+1}(0)>0$. Based on this induction result, starting from $r_1(0)>0$, it can recursively infer that all the initial conditions $r_{j}(0)>0, j=1,\cdots,m$. The solution of  \eqref{obj5}, \eqref{obj6} is
		 $r_m(t)=r_m(0)e^{-c_mt}$, $r_j(t)=r_j(0)e^{-c_jt}+\int_{0}^{t}e^{-c_j(t-y)}r_{j+1}(y)\diff y, j=1,\cdots,m-1$,  together with the above positive initial condition, the lemma is then obtained.
	\end{proof}
	\begin{lemma}\label{lemmasafe2} 
		The high-relative-degree CBFs $z_i(t), i=1,\cdots,n$ are nonnegative under the selection of design parameters  \eqref{ki}, \eqref{k}, i.e., $z_i(t)\ge 0, i=1,\cdots,n$, for time $t\ge D$.
	\end{lemma}
	\begin{proof}
	Setting $t=D$ and replacing $Y(D)$ with $P(0)$ in \eqref{zi}, then one gets
		\begin{align}
				&z_{i+1}(D)=P_{i+1}(0)-h_{i}(\underline{P}_{i}(0),\underline{s}^{(i-1)}(D))-s^{(i)}(D)\notag\\
				&=P_{i+1}(0)-s^{(i)}(D)+k_iz_i(D)+\psi_i(\underline{P}_{i}(0))\notag\\
				&-\sum_{k=1}^{i-1}\big(\frac{\partial h_{i-1}}{\partial P_k}\bigl(P_{k+1}(0)+\psi_k\bigl)+\frac{\partial h_{i-1}}{\partial s^{(k-1)}}s^{(k)}(D)\big).	\label{zd}
		\end{align}
		Under the selection of parameters $k_i$ in \eqref{ki}, \eqref{k}, we also obtain the induction step: if $z_i(D)>0$, then $z_{i+1}(D)>0$ according to \eqref{zd}. Because of the base case that $z_1(D)=P_1(0)-s(D)>0$ ensured by \eqref{eq:assum4} in Assumption \ref{assum4}, applying the induction step obtained above, we thus have that all $z_i(D)>0$, $i=1,\cdots,n$. Considering the structure of \eqref{obj1}, \eqref{obj2}, together with the initial conditions $z_i(D)>0$, and  $w(0,t)=r_1(t-D)\ge 0, t\ge D$ according to Lemma \ref{lemmasafe1} as well as \eqref{obj3}, \eqref{obj4}, we obtain that $z_i(t)> 0, i=1,\cdots,n$ on $t\in[D,\infty]$.\end{proof} 
  Now, we are ready to show the proof of Theorem \ref{theo1}:
  \begin{proof} \noindent\textbf{1)} 
  The stability results of the target system in Lemmas \ref{lem:targetstability} and \ref{lemma:deltarget} indicate that $|Z(t)|^2$, $|R(t)|^2$, $\sum_{i=0}^{m}\|w_x^{(i)}(\cdot,t)\|^2$, $|\delta(x,t)|^2$, $\forall x\in[0,1]$, and $|\delta_t^{(i)}(1,t)|^2,i=0,\cdots,m$ are exponentially convergent to zero.
  
  Considering the tracking error $|z_1(t)|$ exponentially converges to zero, we have that $|y_1(t)-s(t)|$ exponentially converges to zero.

  Recalling Proposition \ref{leminver1}, applying Cauchy-Schwarz inequality, it follows \eqref{inveryi}--\eqref{inverh} and exponential convergence to zero of $z_i(t)^2$ that $y_i(t)^2$ are bounded and the ultimate bound depends on $\sum_{j=0}^{i-1}s^{(j)}(t)^2$ for $i=1,\cdots,n$. According to \eqref{inverker} in Proposition \ref{leminver1}, it is then obtained that  $\|u(\cdot,t)\|^2$ is bounded and the ultimate bound depends on the functions $\int_0^{1}\underline{s}^{(n)}(t+Dx)^2\diff x$. Recalling \eqref{uxt}, we thus have \begin{align}
	\int_{t-D}^t x_1(\tau)^2\diff\tau=\frac{\|u\|^2}{b^2}\label{equx1}
\end{align}
 is also ultimately bounded by a function depending on $\int_0^{1}\underline{s}^{(n)}(t+Dx)^2\diff x$. Plugging the predicted states $p(1,t)$, $\delta(1,t)$ into the inverse transformation \eqref{inveryi}, we have
	\begin{align}
        p_i(1,t)=&\delta_i(1,t)+\bar h_{i-1}(\underline{\delta}_{i-1}(1,t),\underline{s}^{(i-2)}(t+D))\notag\\&+s^{(i-1)}(t+D),~~ i=1,\cdots,n  \label{pixt}
	\end{align}
 recalling the continuous differentiability of $\bar{h}_i$ in \eqref{inverh} and the exponential convergence to zero of  $\delta_i(x,t)$ proven above, we have that  ${e_i}P(t)=p_i(1,t)$ is bounded for $t>0$ and the ultimate bound depends on $\underline{s}^{(i-1)}(t+D)$ for $i=1,\cdots,n$. Recalling \eqref{eq:assum3} in Assumption \ref{assum4} and \eqref{sol}, we know $|p(x,t)|^2$ are bounded and the ultimate bound depends on $\underline{s}^{(i-1)}(t+Dx)^2$. 
Recalling \eqref{inverdelta}, one obtains that $|\underline{\bar\Delta}^{(i)}(t)|^2$ is bounded with the ultimate bound depending on $ \sum_{j=0}^{n}s^{(j)}(t+D)^2$, because $|\delta_t^{(i)}(1,t)|^2,i=0,\cdots,m$ are exponentially convergent to zero and the $\bar{h}_n$ in \eqref{inverh} is continuously differentiable. Recalling Proposition \ref{leminver1}, applying Cauchy-Schwarz inequality, it follows \eqref{inver ri}--\eqref{invertau} and exponential convergence to zero of $r_i(t)^2$, $i=1,\cdots,m$ that $x_i(t)^2$ are bounded and the ultimate bound depends on $\sum_{j=0}^{i-1}s^{(n+j)}(t+D)^2$ for $i=1,\cdots,m$. Therefore, we conclude that $\Psi(t)$ defined by \eqref{eq:Psi} is bounded, and the ultimate bound depends on the given target trajectory $s(t)$.

Next, we show that $\Psi(t)$ \eqref{eq:Psi} is exponentially convergent to zero if the target trajectory $s(t)\equiv 0$. Considering the exponential convergence to zero of $|Z(t)|$ and $\bar{h}_n(0,0)=0$, applying Cauchy-Schwarz inequality for \eqref{inveryi} in Proposition \ref{leminver1},  we know that $|Y(t)|$ is exponentially convergent to zero. Recalling the inverse transformation of $u(x,t)$ \eqref{inverker} in Proposition \ref{leminver1}, it is obtained from the exponential convergence to zero of $|{\delta}(x,t)|$, $\Vert w(\cdot,t) \Vert$, and $\bar{h}_n(0,0)=0$ that $\Vert u(\cdot,t)\Vert$ is exponentially convergent to zero. It follows \eqref{equx1} that $(\int_{t-D}^t x_1(\tau)^2\diff\tau)^{\frac{1}{2}}$ is also exponentially convergent to zero. It is obtained from \eqref{pixt} that $p(1,t)=P(t)$ are exponentially convergent to zero considering  $\bar{h}_i(0,0)=0$ in \eqref{inverh}, and the exponential convergence to zero of  $|\delta(x,t)|$ obtained before. Moreover, according to \eqref{sol} and \eqref{eq:assum3} in Assumption \ref{assum4}, the exponential convergence to zero of $|{p}(x,t)|$, $\forall x\in[0,1]$ is obtained. Then, recalling \eqref{inverdelta}, we have that $|\bar\Delta^{(i)}(t)|, i=0,\cdots,m$ are exponentially convergent to zero because of the exponential convergence to zero of  $|\delta_t^{(i)}(1,t)|$ obtained above and continuous differentiability of $\bar{h}_n$.
Finally, it is obtained from \eqref{inver ri}--\eqref{invertau} in Proposition \ref{leminver1} that $|X(t)|$ is exponential convergent to zero, recalling the exponential convergence to zero of $|R(t)|$, $|\underline{\bar\Delta}^{(i)}(t)|, i=0,\cdots,m$, and $\bar{\tau}_i(0,0)=0$.

The property 1 is obtained.
  
\noindent\textbf{2)} Over the time period $t\in[0,D]$, i.e, when no control action reaches, the $Y$-system \eqref{equ1}, \eqref{equ2} is only actuated by the initial input history signal $x_1(t)$, $t\in[-D,0]$.
We know from \eqref{eq:assum4} in Assumption \ref{assum4} and \eqref{eq:preP} that 
	\begin{align}
		y_1(t)-s(t)=P_1(t-D)-s(t)>0,t\in[0,D],\label{zd0}
	\end{align}
	i.e., the safety is ensured on $t\in[0,D]$. 
  From Lemma \ref{lemmasafe2}, it holds that 
$	
		z_1(t)=y_1(t)-s(t)\ge 0, t\in[D,\infty).
$
	Consequently, the safety is ensured all the time.  The property 2 is thus proved.\par
	The proof of Theorem \ref{theo1} is complete.
\end{proof}	
Based on safe infinite-dimensional backstepping transformations, a nominal safe delay-compensated controller is proposed in this section. Next, given the uncertainty of the delay, we will propose a safe delay-adaptive controller.

\section{Safe Delay-Adaptive Controller}\label{secadaptive}
\subsection{Delay-adaptive control design}\label{identifier}
Following the delay-adaptive design in \cite{delay9}, we obtain the delay identifier:
\begin{align}
		&\hat{D}(t_{i+1})=\arg\min\Big\{|\ell-\hat{D}(t_{i})|^2: \ell\in D_0,\notag\\
		&G_n(t_{i+1},\mu_{i+1})\ell=F_n(t_{i+1},\mu_{i+1}),n=1,2,\cdots\Big\},\label{iden}
\end{align}
where the set $D_0:=\{\ell\in \mathbb{R}:\underline{D}\leq \ell\leq\overline{D}\}$ uses the known bounds of the unknown delay given in Assumption \ref{assum2}, and where  ${\{t_i\ge0\}}_{i=0}^\infty,i=0,1,2,\cdots$ is a sequence of time instants for identification, defined as
\begin{equation}
    t_{i+1}=t_i+T.
\end{equation}
The free positive design parameter $T$ denotes the dwell time between two adjacent triggering times. { A larger $T$ improves the robustness of the delay identifier to sensor measurement errors but prolongs parameter identification time. Conversely, a smaller $T$ enables faster identification of unknown delay but may reduce robustness to sensor measurement error.} The time instant $\mu_{i+1}$ is defined as
\begin{equation}
	\mu_{i+1}:=\min\{t_g:g\in\{0,\cdots,i\},t_g\ge t_{i+1}-\tilde{N}T\},\label{mu}
\end{equation}
where the positive integer $\tilde{N}\ge 1$ is a free design parameter, 
which determines the size of the data set used in the delay identification at $t_{i+1}$, reflecting a trade-off between the identifier's robustness and computation cost.

The functions $G_n, F_n$ in \eqref{iden} are given as
$
        G_n(t_{i+1},\mu_{i+1})=\int_{\mu_{i+1}}^{t_{i+1}}g_n(t)^2\diff t$,
$	F_n(t_{i+1},\mu_{i+1})=\int_{\mu_{i+1}}^{t_{i+1}}g_n(t)f_n(t)\diff t,
$
where $f_n(t)=\pi  n\int_{0}^{t}\int_{0}^{1}\cos(x\pi n)u(x,\tau)\diff x\diff\tau$,
$g_n(t)=-\int_{0}^{1}\sin(x\pi n)u(x,t)\diff x$. 
The detailed design process and the proof of exact delay identification can be found in \cite{delay9}.  
We only consider the scenario where $x_1(t-D)=0$ on $t\in[0,D)$, i.e., there is no signal reaching $Y$-subsystem before $t=D$ according to \eqref{equ2}, in the adaptive control. For the case that $x_1(t)$ is not identically zero on $t\in[-D,0)$, a slight modification is needed in the formulation of the delay identifier, and some expended analysis is required in the proof.

Now, using  the proposed estimate $\hat{D}(t_{i})$ to replace the unknown delay $D$ in the nominal controller \eqref{u}, we construct a delay-adaptive controller
\begin{align}
	U_d(t)=\mathcal{U}(t,\hat{D}(t_i)), ~ t\in[t_i,t_{i+1}). \label{adacon}
\end{align}
The safety ensured by the nominal control can not be guaranteed here because of the delay identification error.
Following the safe-adaptive control design in \cite{safe5}, we introduce a QP safety filter \eqref{ua} to override the potentially unsafe adaptive controller \eqref{adacon} to enforce the safety in the adaptive control.
\subsection{Safe delay-adaptive control design}
First, considering the unknown $D$, we select the design parameters $k_i$, $c_j$ \eqref{ki}, \eqref{ci} as 
\begin{align}
	k_i&>\max_{\mathcal{D}\in D_0}\{2,\check{k}_i(\mathcal{D})\}, i=1,\cdots,n-1, k_n>1, \label{kia}\\
	c_j&>\max_{\mathcal{D}\in D_0}\{2,\check{c}_j(\mathcal{D})\}, j=1,\cdots,m-1, c_m>1, \label{cja}
\end{align}
where  $\check{k}_i(\mathcal{D}),\check{c}_j(\mathcal{D})$ are obtained by replacing the unknown delay $D$ in \eqref{k}, \eqref{c} by $\mathcal{D}\in[\underline{D},\overline{D}]$, where the bounds $\underline{D},\overline{D}$ are known according to Assumption \ref{assum2}. Because the condition \eqref{kia}, \eqref{cja} is a subset of the one \eqref{ki}, \eqref{ci}, the positive initialization about $r_j(0)$ and $z_i(D)$, as shown in the proofs of Lemmas \ref{lemmasafe1}, \ref{lemmasafe2}, still hold here. Then recalling the target system \eqref{obj1}--\eqref{obj6} and the analysis about the safety in the proof of Property 2 in Theorem \ref{theo1}, we know the safety objective $z_1(t)=y_1(t)-s(t)\ge0$ is achieved as long as $r_m(t)>0$ for all time $t>0$, of which a sufficient condition is
\begin{equation}
	\dot{r}_m(t)\ge-\overline{c}r_m(t),\label{safecon1}
\end{equation}
where the positive parameter $\overline{c}$ is free. A safe region for the control action is then obtained from \eqref{safecon1} as $\mathcal{S}(t)=\{u\in\mathbb{R}:bu\ge b\mathcal{U}^*(t,D)\}$ where
\begin{equation}
	\mathcal{U}^*(t,D)=\frac{1}{b}\big[(c_m-\overline{c})r_m(t)+\tau_m+\Delta^{(m)}(t)\big].\label{uc}
\end{equation}
Considering the unknown $D$, by replacing the unknown delay by $\mathcal{D}$, a conservative safe region of the adaptive control input is introduced as 
\begin{equation}
	C(t)=\left\{u\in \mathbb{R}:bu\ge \max_{\mathcal{D}\in D_0}b\mathcal{U}^*(t,\mathcal{D})\right\}.\label{saferegion}
\end{equation}
By using a QP safety filter to constrain the input signal within this safe region \eqref{saferegion} before exact identification is achieved, we build the following safe adaptive controller:
\begin{equation}
	U_a(t)=\begin{cases}
		\mathop{\arg\min}\limits_{u\in\mathbb{R}} \{ |u-U_d|\}^2\\
		s.t.\quad u\in C(t), &t\in[0,t_f)\\
		U_d,        &t\in[t_f,\infty)
	\end{cases}
\end{equation}
whose explicit solution is
\begin{equation}
U_a(t)=\begin{cases}
		\begin{cases}
		\max\{U_d,\max\limits_{\mathcal{D}\in D_0}\mathcal{U}^*(t,\mathcal{D})\},& \text{if}\quad b>0;\\
	     \min\{U_d,\min\limits_{\mathcal{D}\in D_0}\mathcal{U}^*(t,\mathcal{D})\},& \text{if}\quad b<0,
		\end{cases}\\
	\phantom{U_d,\max\limits_{\mathcal{D}\in C_i}U^*(t,\mathcal{D})\}=====}0\le t\le t_f\\
	 U_d(t). \phantom{aaaaaaaaaaa=====}  t> t_f
	\end{cases}\label{ua}
\end{equation}

The switching time $t_f$ is the triggering time when the delay $D$ is exactly identified, determined by
\begin{equation}
	t_f=\min\{t_i:\exists t\in [0,t_i),u(\cdot,t)\neq0\},
\end{equation}
according to the proof of exact identification of delay shown in Sec. V of \cite{delay9}.
The diagram of the proposed safe delay-adaptive control system is depicted in Fig. \ref{fig dig}. The practical implementation of this safe delay-adaptive control law can refer to Remarks \ref{remarkimple}, \ref{remarkdelay} in the simulation. 
\subsection{Result with safe delay-adaptive control}
\begin{figure}[t]
	\centering	
	\includegraphics[width=1\linewidth, height=0.2\textheight]{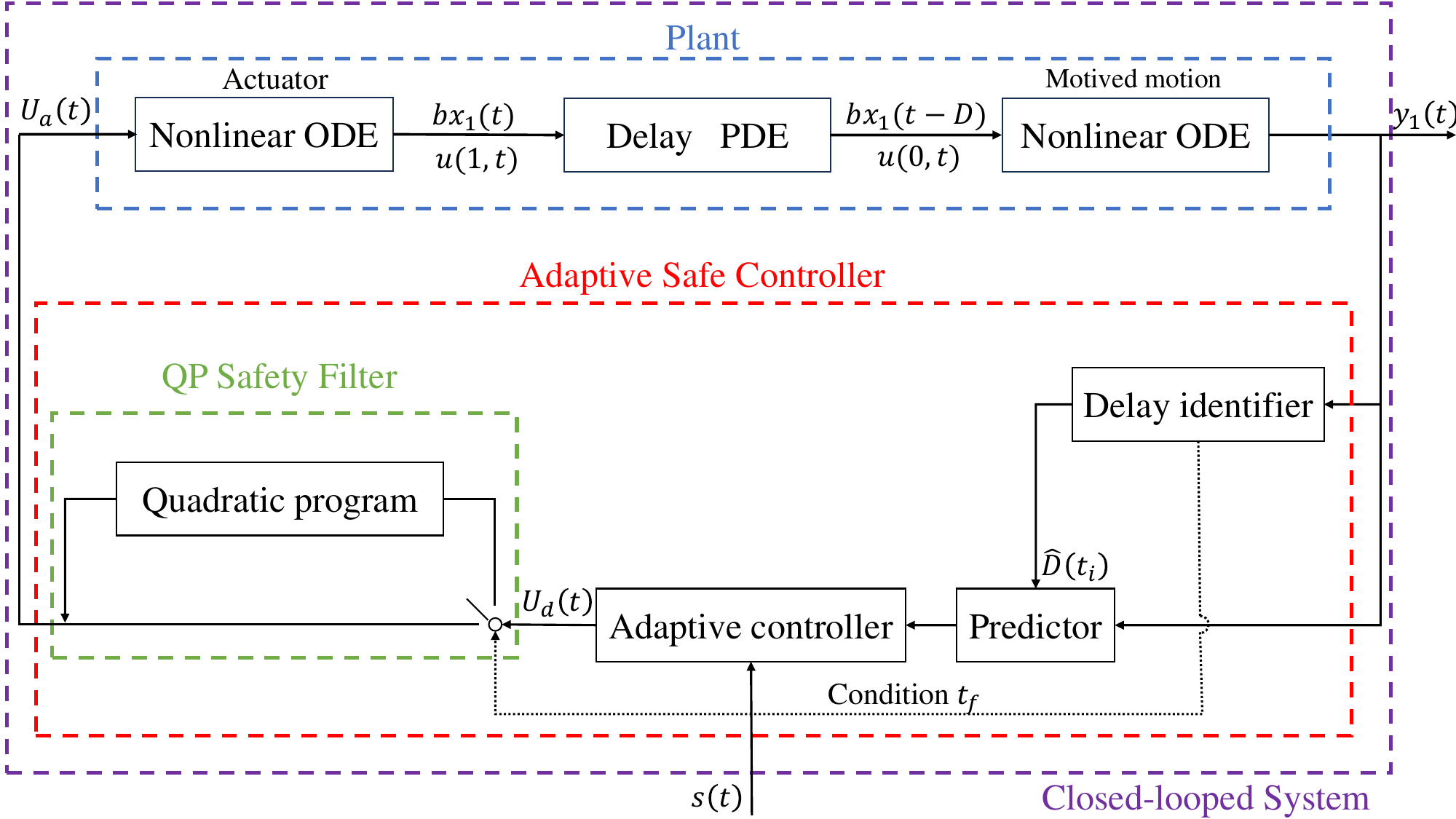}
	\caption{ The diagram of the safe delay-adaptive control.}		
	\label{fig dig}
\end{figure}
Comparing the safe nominal controller \eqref{u} with the safe adaptive controller \eqref{ua}, we define their difference as
	\begin{equation}
	\gamma(t)=bU_a(t)-bU(t).\label{gam}
	\end{equation}
 Then, implementing $U_a(t)$ as the input into the original system \eqref{equ1}--\eqref{equ4}, the target system becomes \eqref{obj1}--\eqref{obj5} with
	\begin{equation}
	\dot{r}_m(t)=-c_mr_m(t)+\gamma(t). \label{gam1}
\end{equation}
\begin{remark}\label{remark2}
	In the control input \eqref{u}, the delay $D$, which exists in the predictor state $P_i(t)$ of the distal $Y$-system as shown in \eqref{pred1}, is associated with $\Delta(t)$, i.e., the parts related to states of the $Y$-system, while independent of the signals from the $X$-system. Thus, the function $\gamma(t)$ given by \eqref{gam} does not contain the signals from the $R$-system (in the form of states of the target system), that is, $\gamma(t)$ can be regarded as an external signal to \eqref{gam1}.
\end{remark}
\begin{proposition}\label{pro1}
	For every $(u(\cdot~,0),X(0),Y(0))\in \mathbb{C}^{m-1}\allowbreak([0,1]) \times\mathbb{R}^m\times\mathbb{R}^n$, there exist a unique solution $(u,X,Y)\allowbreak\in \mathbb{C}^{m-1}([0,\infty)\times[0,1])\times\mathbb{C}^0([0,\infty);\mathbb{R}^m)\times\mathbb{C}^0([0,\infty);\mathbb{R}^n)$ to the system \eqref{equ1}--\eqref{equ4} with control input \eqref{ua}.
\end{proposition}
\begin{proof}
It is obtained from $(u(\cdot~,0),X(0),Y(0))\in \mathbb{C}^{m-1}\allowbreak([0,1])\times\mathbb{R}^m\times\mathbb{R}^n$ and the transformations \eqref{zi}--\eqref{hi}, \eqref{ker}, \eqref{ri}--\eqref{tau} that $(w(\cdot,0),R(0),Z(0))\in \mathbb{C}^{m-1}([0,1])\times\mathbb{R}^m\times\mathbb{R}^n$. According to Remark \ref{remark2} and [Proposition 1, \cite{safe5}], for the target system consisting of \eqref{obj1}--\eqref{obj5} and \eqref{gam1}, we have that $(w(\cdot,t),R(t),Z(t))\in \mathbb{C}^{m-1}([0,\infty)\times[0,1])\times\mathbb{C}^0([0,\infty);\mathbb{R}^m)\times\mathbb{C}^0([0,\infty);\mathbb{R}^n)$ in the weak sense. Recalling the inverse transformations \eqref{inveryi}--\eqref{invertau}, this proposition is then obtained.
\end{proof}
The result of safe delay-adaptive control is presented as follows. 
\begin{theorem}\label{theo2}
	For the initial condition $Y(0)\in\mathbb{R}^n,X(0)\in\mathbb{R}^m$ satisfying Assumptions \ref{assum4}, \ref{assum5}, the history input signal $x_1(t)=0$ on $t\in[-D,0)$, and a target trajectory $s(t)$ satisfying Assumption \ref{as:sa}, choosing the design parameters $k_1,\cdots,k_n,c_1,\cdots,c_m$ satisfying \eqref{kia}, \eqref{cja}, the closed-loop system \eqref{equ1}--\eqref{equ4} with the safe delay-adaptive controller \eqref{ua} has the following properties
	\begin{enumerate}
		\item The delay estimation $\hat{D}(t)$ is bounded and reaches the true value in finite time $t_f$.
		\item The output $y_1(t)$ exponentially tracks the target trajectory $s(t)$ in the sense that $|y_1(t)-s(t)|$ exponentially converges to zero, and all plant states, i.e, $\Psi(t)$ given by \eqref{eq:Psi}, are bounded, and the ultimate bound depends on the target trajectory. If the target trajectory $s(t)\equiv0$, then $\Psi(t)$ is exponentially convergent to zero.
		\item {The safety is ensured in the sense that $y_1(t)-s(t)\geq0$ hold on $t\geq 0$}. 
	\end{enumerate}
\end{theorem}
\begin{proof}
	
	\textbf{1)} The proof of Property 1 can be found in Sec. \uppercase\expandafter{\romannumeral 4}-C of \cite{delay9}.\par
	\textbf{2)} Recalling the Lyapunov function $V(t)$ \eqref{v}, choose the analysis parameters as $a_0>\overline{D},a_i>0, i=1,\cdots,m$, i.e., replacing the unknown $D$ in the condition \eqref{ai} by the known bounds $\underline{D},\overline{D}$. Recalling \eqref{gam1} in the target system of the adaptive case, the inequity \eqref{bi} now becomes $
\sum_{i=1}^{m}a_ieD^{i-1} r_1^{(i)}(t)^2\leq\sum_{i=1}^{m}b_i r_i(t)^2+\bar{b}\gamma^2(t),$
where $b_i,\bar{b}$ depend on the  upper bound $\overline{D}$, the design parameters $c_i$ in \eqref{cja}, and the analysis parameters $a_i$. Implementing the process similar to \eqref{v1}--\eqref{lya} and choosing the design parameters $\rho$ \eqref{rho} as ${\rho}>\max\left\{ \frac{a_0e}{3\underline{D}}+\frac{1}{3}b_1,b_m,\frac{1}{2}b_j \right\}+1, j=2,\cdots,m-1$. The time derivative of $V(t)$ in \eqref{eq:Ly1} becomes
	\begin{align}
		\dot V(t)\leq-{\varrho} V(t)+\bar{b}\gamma^2(t)+{\rho} r_m(t)\gamma(t), \label{dv}
\end{align}
where ${\varrho}=\frac{1}{\theta_2}\min\left\{1, \frac{a_i}{2\overline{D}} \right\}, i=0,\cdots,m$.
For the time period $t\in[t_f,\infty]$, recalling property 1 in Theorem \ref{theo2} and \eqref{ua}, we know that $U_a(t)=U(t)$, i.e., $\gamma(t)=0$. Thus, for \eqref{dv}, we obtain 
$ V(t)\leq V(t_f)e^{-{\varrho}(t-t_f)},~ \forall t \ge t_f$.
From Remark \ref{remark2} and transformations \eqref{zi}--\eqref{hi}, \eqref{ker}--\eqref{delta}, one obtains  $\gamma^2(t)\le\Upsilon\Omega(t)$ for some positive $\Upsilon$, where $\Omega(t)$ is defined in \eqref{ome1}. Thus it follows \eqref{ome2}, \eqref{dv} that $\dot V(t)\leq-{\varrho} V(t)+\varrho_0V(t), t\in[0,t_f)$ for some positive $\varrho_0$. We then get that $V(t)\leq V(0)e^{|\varrho_0-{\varrho}|t}, t\in[0,t_f)$. Considering the continuity of $V(t)$ by recalling Proposition \ref{pro1}, we have that
	$
	V(t_f)\leq V(0)e^{|\varrho_0-{\varrho}|t_f}.
$	
From the above relation, it further implies that
\begin{equation}
	V(t)\leq V(0)e^{|\varrho_0-{\varrho}|t_f+{\varrho}t_f}e^{-{\varrho}t}
\end{equation}
for $t\in[0,\infty)$.
According to \eqref{ome3}, we have that $\Omega(t)\leq\frac{\theta_2}{\theta_1}\Upsilon_0\Omega(0)e^{-\varrho t},$ where $\Upsilon_0=e^{|\varrho_0-{\varrho}|t_f+{\varrho}t_f}$. Through the following process in the proof of property 1 in Theorem \ref{theo1}, we obtain the property 2 in this theorem. \par

\textbf{3)} Implementing $U_a(t)$ as the input into the original system \eqref{equ1}--\eqref{equ4} and recalling the transformation \eqref{zi}--\eqref{hi}, \eqref{ri}--\eqref{tau} , we have
 \begin{equation}
	 \dot{r}_m(t)=-\bar{c}r_m(t)+\overline{\gamma}(t),~t\in[0,t_f) \label{gam2}
	\end{equation} 
where $\overline{\gamma}(t)=bU_a(t)-b\mathcal{U}^*(t,D)\ge0$ recalling \eqref{u}, \eqref{safecon1}--\eqref{ua}. Since $z_i(D),r_j(0)$ are positive under the choice of the design parameters $k_i,c_j$ in \eqref{kia}, \eqref{cja}, we have that $r_m(t)>0, t\in[0,t_f)$ from the structure of \eqref{gam2}. Considering the continuity of $r_m(t)$ in Proposition \ref{pro1}, we obtain $r_m(t_f)=r_m({t_f}^-)>0$. When the adaptive input signal $U_a(t)$ is equal to $U_d(t)$ at the moment $t=t_f$, the solution $r_m(t)$ is $r_m(t)=r_m(t_f)e^{-c_m(t-t_f)}$. Thus, the  nonnegativity of $r_m(t)$ can be ensured during $t\in[0,\infty)$. It implies that all states $r_j(t), z_i(t)$ are all nonnegative from the same process in the proof of Lemmas \ref{lemmasafe1},\ref{lemmasafe2} and property 2 in Theorem \ref{theo1}. The safety, i.e., $z_1(t)=y_1(t)-s(t)\ge0,t>0$, is then proved.

The proof of this theorem is complete. 
\end{proof}

\section{Application in Safe Vehicle Platooning}\label{secexample}
\begin{figure}[t]
	\centering	
	\includegraphics[width=1\linewidth, height=0.2\textheight]{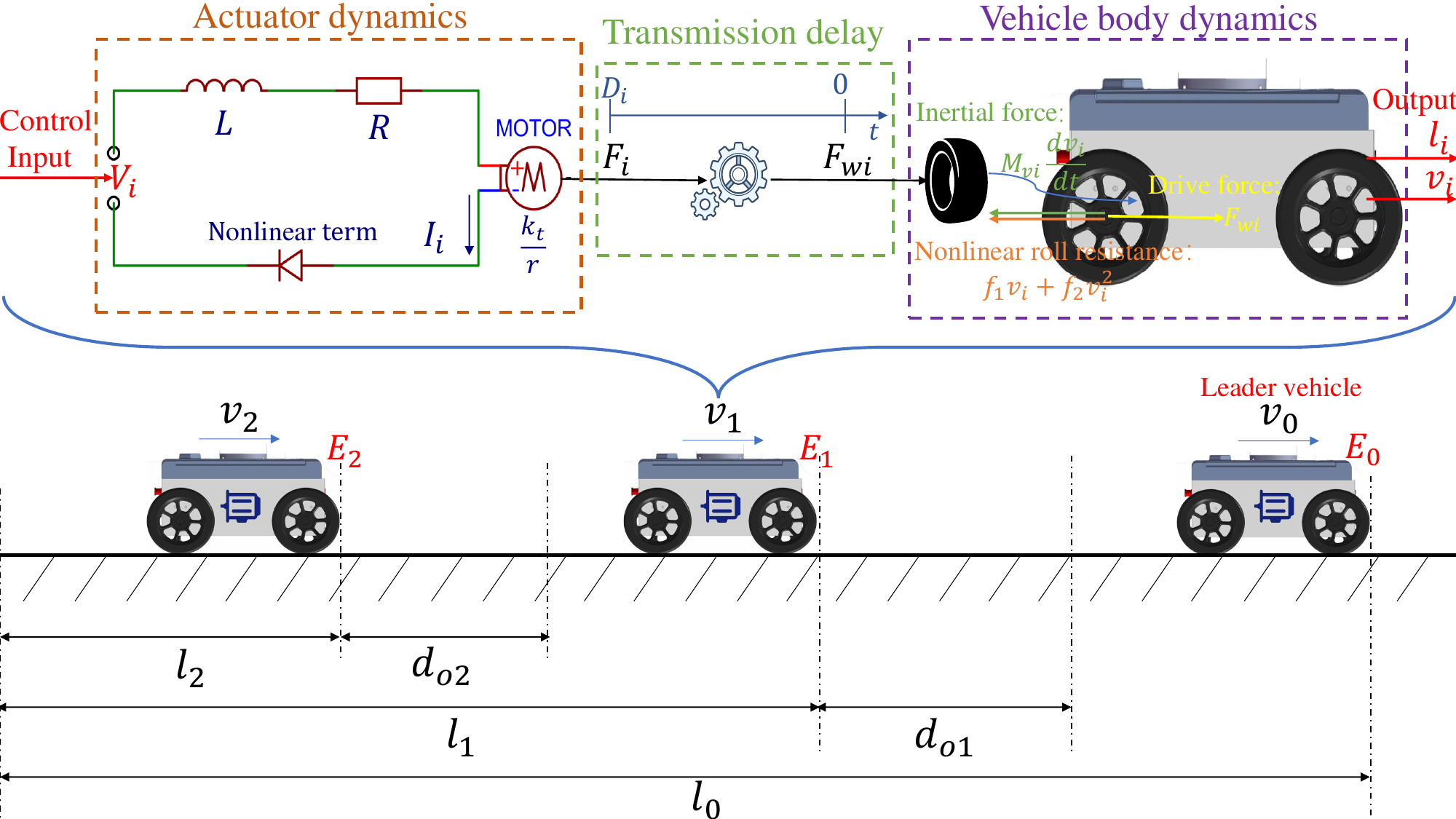}
	\caption{ Vehicle platooning with leader $E_0$, and the followers $E_i,i=1,2$, where the safe distances to be maintained are $d_{oi}$.}		
	\label{fig：0}
\end{figure}
\subsection{Physical model}
In this section, we validate the effectiveness of the designed controller in the practical application of safe vehicle platooning by simulations. Considering the scenario described in Fig. \ref{fig：0}: an electric vehicle $E_0$ is the leader of the vehicle string and is traveling at speed $v_0(t)$. The vehicle $E_1$ is running behind the vehicle $E_0$ at speed $v_1(t)$, and the third vehicle $E_2$ is traveling behind the vehicle $E_1$ at speed $v_{2}(t)$. The control task is to achieve a safe formation of these three vehicles while avoiding collisions, i.e., regulating the distances between adjacent vehicles to converge to the safe distances, denoted by $d_{o1}$, $d_{o2}$, but not to breach the safe distances all the time, i.e., $d_1(t)=l_0(t)-l_1(t)\geq d_{o1}$, $d_2(t)=l_1(t)-l_2(t)\geq d_{o2}$ and $d_1(t)$, $d_2(t)$ converge to $d_{o1},d_{o2}$, respectively. 

In $i$-th electric vehicle, the control input is the voltage of the electric motor, shown in Fig. \ref{fig：0}, whose dynamics are described by nonlinear ODEs:
\begin{equation}
	F_{i}=\frac{k_t}{r}I_i,\quad\dot{I_i}=-\frac{R}{L}I_i-a{I_i}^2+\frac{1}{L}\mathcal{V}_i,\label{veh3}
\end{equation}
where $k_t$ is the torque constant of the DC motor, $r$ is the length of moment arm,  $I_i$  is the motor current, $R$ is the resistance of the motor, $L$ is the inductance and $\mathcal{V}_i$ is the input voltage. To reduce the modeling error between the mathematical model and the practical model, we introduce an unmodeled nonlinear term $aI_i^2$ to approximate the nonlinear elements in the drive circuit where the coefficient $a$ is to be calibrated in practice by matching the mathematical and practical models. The output force $F_i$ of the motor is transmitted to the wheel, generating the wheel drive force ${F_w}_{i}$, through a set of transmissions, like a gearbox. There always exists a delay $D_i$, whose length is not easy to know exactly in advance, in such the transmission between the motor and the wheel, i.e.,
\begin{align}
{F_w}_{i}(t)=F_i(t-D_i).\label{eq:delaysim}
\end{align}
Based on the wheel drive forces ${F_w}_{i}$, according to \cite{acc1,bib25}, the dynamic of $i$-th vehicle is modeled as the following nonlinear ODE:
\begin{align}
	M_{vi}\frac{dv_i}{dt}=F_{wi}-f_1v_i-f_2{v_i}^2,\label{veh1}
\end{align}	
where $M_{vi}$ is the vehicle mass, and $f_2{v_i}^2$ describes the nonlinear damping force. The physical parameters used in the simulation are given in Tab. \ref{tab1}.
\subsection{Matching the physical model and the plant \eqref{equ1}--\eqref{equ4}}
 For $i$-th vehicle, setting $y_{i1}(t)=-l_i(t),y_{i2}(t)=-v_i(t)$, $x_{i1}(t)=F_{i}$ (i.e., $x_{i1}(t-D_i)=F_{wi}$), and $U_i=\frac{k_t}{rL}\mathcal{V}_i$, the physical model \eqref{veh3}--\eqref{veh1} is rewritten as
\begin{align}
	\dot y_{i1}(t)&=y_{i2}(t),    \\
	\dot y_{i2}(t)&=\frac{-f_1y_{i2}(t)+f_2{y_{i2}(t)}^2}{M_{vi}}-\frac{x_{i1}(t-D)}{M_{vi}},\\
 \dot x_{i1}(t)&=-\frac{R}{L}x_{i1}-a\frac{r}{k_t}x_{i1}^2+U_i(t),
\end{align}
which is covered by the considered plant (with the states $Y_i(t)=[y_{i1}(t),y_{i2}(t)]^T, X_i(t)=x_{i1}(t)$, where $i$ denotes $i$-th considered plants).
Considering the physical safety constraint $l_{i-1}(t)-l_i(t)-d_{oi}\geq0 $, we set the target trajectory of the $i$-th vehicle as
$s_i(t)=-l_{i-1}(t) +d_{oi}$, where $l_{i-1}(t)$ is the measurable displacement of the front vehicle.
We set that the initial position of the leader vehicle $E_0$ is $l_0(0)=10$, with the speed given by $v_1(t)=4+\sin(t)$ and thus the target trajectory for $E_1$ is $s_1(t)=-4t+\cos(t)-11+d_{o1}$. The target trajectory for vehicle $E_2$ is $s_2(t)=y_{11}(t)+d_{o2}$. The unknown delays of vehicles $E_1,E_2$ are given in Tab.\ref{tab1}, with known bounds as $\underline{D}=0.2$, $\overline{D}=4$. We only deal with the case $D_1\ge D_2$ in this application, considering that the controller of the vehicle $E_2$ requires the predicted value of the displacement $y_{11}$ of the vehicle $E_1$ in a time horizon of $D_2$, while the predictor time span in $E_1$ controller is $D_1$. 
Assuming the initial values of vehicles' speed and distance between them as $v_1(0)=1$, $v_2(0)=2$, $d_1(0)=5$, $d_2(0)=5$, then the initial conditions of $Y_i(t)$  are given as $y_{11}(0)=-5$, $y_{12}(0)=-1$, $y_{21}(0)=0$, $y_{22}(0)=-2$. Besides, we take $x_{11}(0)=2$,$x_{21}(0)=1$ with $x_{11}(t)\equiv0$, $x_{21}(t)\equiv0$ during the period $t\in[-D,0)$. They satisfy Assumptions  \ref{assum4}, \ref{assum5} regarding the initial conditions. 
\begin{table}[tb]
	\renewcommand{\arraystretch}{1.05}
	\centering\caption{Physical parameters in the vehicle platooning model}\label{tab1}
	\begin{tabular}{>{\scriptsize}p{22em}<{\raggedright} >{\scriptsize}p{3em}<{\centering}}
		\hline\\[-2.9mm]  
		\normalsize Parameters(units)& \normalsize Values\\	\\[-2.9mm]\hline
		Linear damping coefficient: $f_1$ ($\frac{N\cdot s}{m}$) & 5  \\
		Aerodynamic drag coefficient $f_2$ ($\frac{N\cdot s^2 }{m}$)& 0.25  \\
		Vehicle mass: $M_{vi}$ ($kg$) & 4     \\
		Target safe distance between $E_0$ and $E_1$: $d_{o1}$ ($m$) & 0.5 \\
		Target safe distance between $E_1$ and $E_2$: $d_{o2}$ ($m$) & 0.5 \\
		Transmission delay in $E_1$: $D_1$ ($s$) & 2.5  \\
		Transmission delay in $E_2$: $D_2$ ($s$) & 1.5  \\
		Torque constant of the DC motor: $k_t$ ($N\cdot m/A$)& 0.8  \\
		Length of moment arm: $r$ ($m$)	& 0.1		\\
		Resistance of the motor: $R$ ($\Omega$) & 5  \\
		Inductance of motor drive circuit: $L$ ($H$) & 0.05\\
      Calibrated coefficient of drive circuit nonlinearity: a& 1\\
		\hline
	\end{tabular}
\end{table}
\subsection{Controller}
 According to \eqref{u}, the nominal controller for this three-order system is 
\begin{align}
		U_i(t)=&-c_{i1}x_{i1}-\varphi_{i1}+\frac{c_{i1}}b\Delta_i(t)+\frac1b\Delta_i^{(1)}(t),\label{ut}
\end{align}
where
$\Delta_i(t)=-k_{i1}k_{i2}P_{i1}-(k_{i1}+k_{i2})(P_{i2}+\psi_{i1}({P}_{i1}))
	+k_{i1}k_{i2}s_{i}(t+D)+(k_{i1}+k_{i2}){s_{i}}^{(1)}(t+D)
		+{s_{i}}^{(2)}(t+D)-{\psi_{i1}}^{(1)}(P_{i1})-\psi_{i2}(\underline{P_i}_{2})$.
Considering the unknown delay, according to \eqref{ua}, the safe delay-adaptive controller is derived from the nominal safe delay-compensated controller \eqref{ut} based on the control design in Sec. \ref{secadaptive}.
According to \eqref{ki}--\eqref{c}, \eqref{kia}, \eqref{cja}, the design parameters for the nominal or adaptive controller are chosen as $k_{11}=k_{21}=3,k_{12}=k_{22}=2,c_{11}=c_{12}=2$. The initial value of $\hat{D}(t_i)$ in \eqref{iden} is defined as $\hat{D}(t_0)=0.2$. The delay-adaptive controller $U_d(t)$ \eqref{adacon} is constructed by replacing the unknown delay with the estimate $\hat{D}(t_i)$ in \eqref{ut}, and the safety filter is built by choosing $\bar{c}=2$ in \eqref{uc}, where it is required to seek the maximum or minimum of the signal  $\mathcal{U}^*(t,\mathcal{D})$ with respect to the delay variable $\mathcal{D}$. The implementation of seeking the maximum or minimum of  $\mathcal{U}^*(t,\mathcal{D})$ and that of the delay estimator are described in the following two remarks.
\begin{remark}[Seeking the maximum or minimum of $\mathcal{U}^*(t,\mathcal{D})$]\label{remarkimple}
	 We divide the known range $[\underline{D},\overline{D}]$ by the interval of $\diff_D$, i,e, each possible delay $\mathcal{D}(i)=\underline{D}+i\diff_D,i=0,\cdots,N_D$, where the number $N_D$ is a freely chosen positive integer (the larger $N_D$ is accompanied with higher accuracy but larger computation source). The predictor values under all possible delays $\mathcal{D}(i)$, $i=0,1,\cdots, N_D$ should be computed for each moment before $t_f$. We set a $N_D$-row data matrix to record the predictor values under all $\mathcal{D}(i)$. As mentioned in Remark \ref{remark1}, taking the same discretization step $d$,  there are $\lceil\frac{\mathcal{D}(i)}{d}\rceil$ predictor values in the row corresponding to the delay $\mathcal{D}(i)$, in the process of computing the prediction for $\mathcal{D}(i)$ (where $\lceil x \rceil$ is defined as the ceiling function: taking the least integer that is greater than or equal to $x$). By taking all the possible predictor values in \eqref{ua} and selecting the maximum or minimum value as a control input, it is ensured that the current control input is in a subset of the safe control region of the nominal safe controller. After $t_f$, i.e. the time when the exact delay identification is achieved, only the values in the row corresponding to the certain delay $\mathcal{D}(i)$ that is closest to the output of the delay estimator are reserved, to continue calculating the predictor values and the control signal along this delay. Then there is no need to calculate the predicted values of other rows anymore after $t_f$, and the control input stays in the original safe control region of the nominal control. 
\end{remark}
Besides, when seeking the maximum or minimum of $\mathcal{U}^*(t,\mathcal{D})$ for vehicle $E_2$, the predicted displacements of vehicle $E_1$ under all possible delays are required to join in this seeking process. Once the identification time $t_f$ is reached, only the predictor value $P_1(t-\hat{D}_1(t_f)+\hat{D}_2(t_f))$ based on $\hat{D}_1(t_f)$ is used in the design of $U_a(t)$ \eqref{ua}.
\begin{remark}[Implementation of the delay estimator]\label{remarkdelay}
The delay estimate $\hat{D}(t)$ is constructed using the finite difference method to approximate the integration with respect to the space variable according to Sec. \ref{identifier}, which may cause some errors between the identified value and true value. The smaller space step $\diff x$ and larger design parameter $T$ can reduce this approximation error, but they also increase the computational burden and computation time. If the difference between the estimates from the identifier at two adjacent updating times is smaller than $2\%$ of the true value, we consider that the approximation error causes this difference, and thus, we keep the estimated value the same as it was at the previous updating time. Additionally, we set an upper limit of $\overline{n}=3$ for $n$ in \eqref{iden} to save computation time for estimation. Other design parameters in the estimator \eqref{iden}--\eqref{mu} are selected as $\tilde{N}=5, T=3$.   
\end{remark}

\begin{figure}[!t]
	\centering	
	\includegraphics[width=0.9\linewidth, height=0.23\textheight]{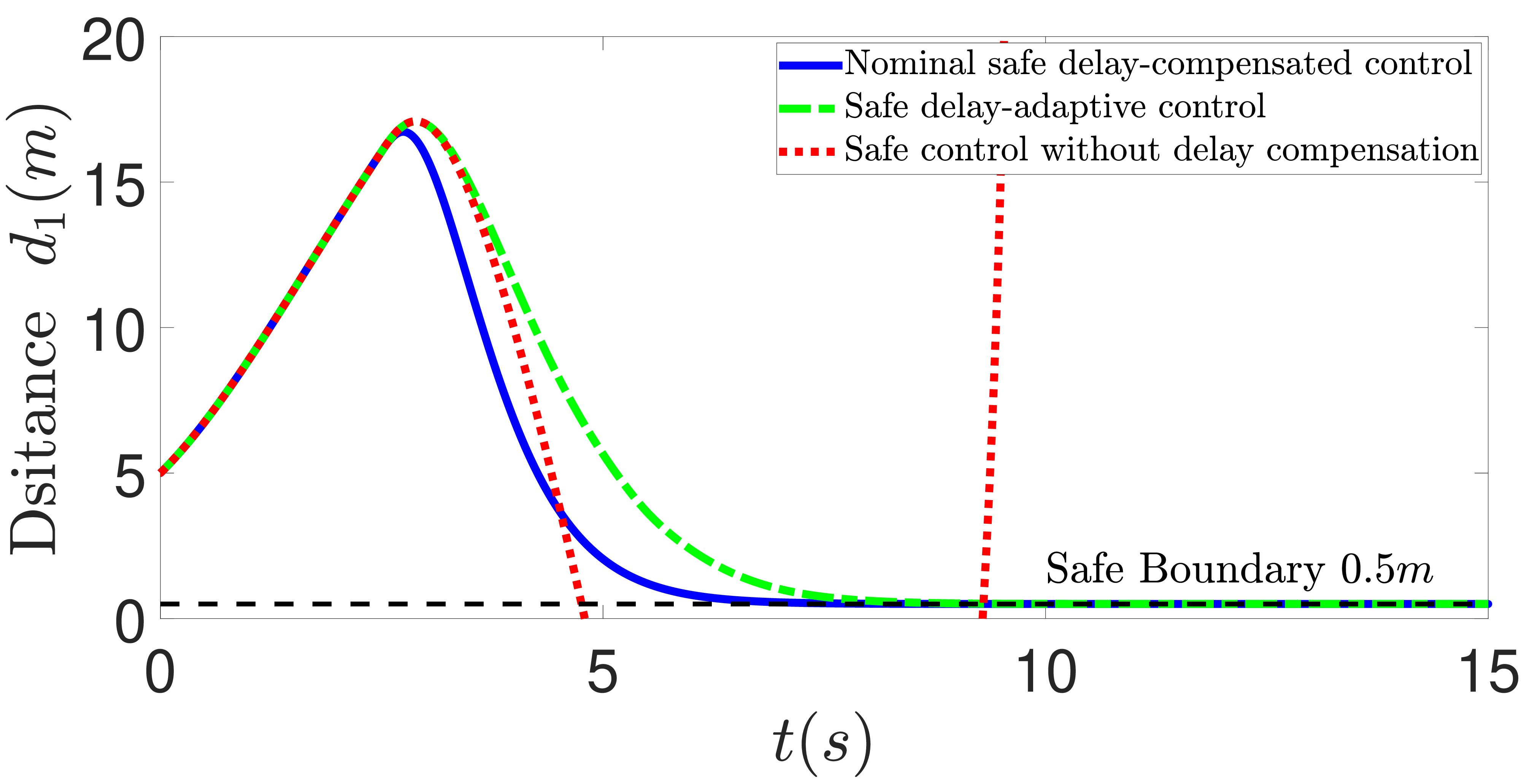}
	\caption{Results for the distance between vehicle $E_0$ and $E_1$, i.e., $d_1(t)=l_0(t)+y_{11}(t)$.}		
	\label{fig:1}
	\includegraphics[width=0.9\linewidth, height=0.23\textheight]{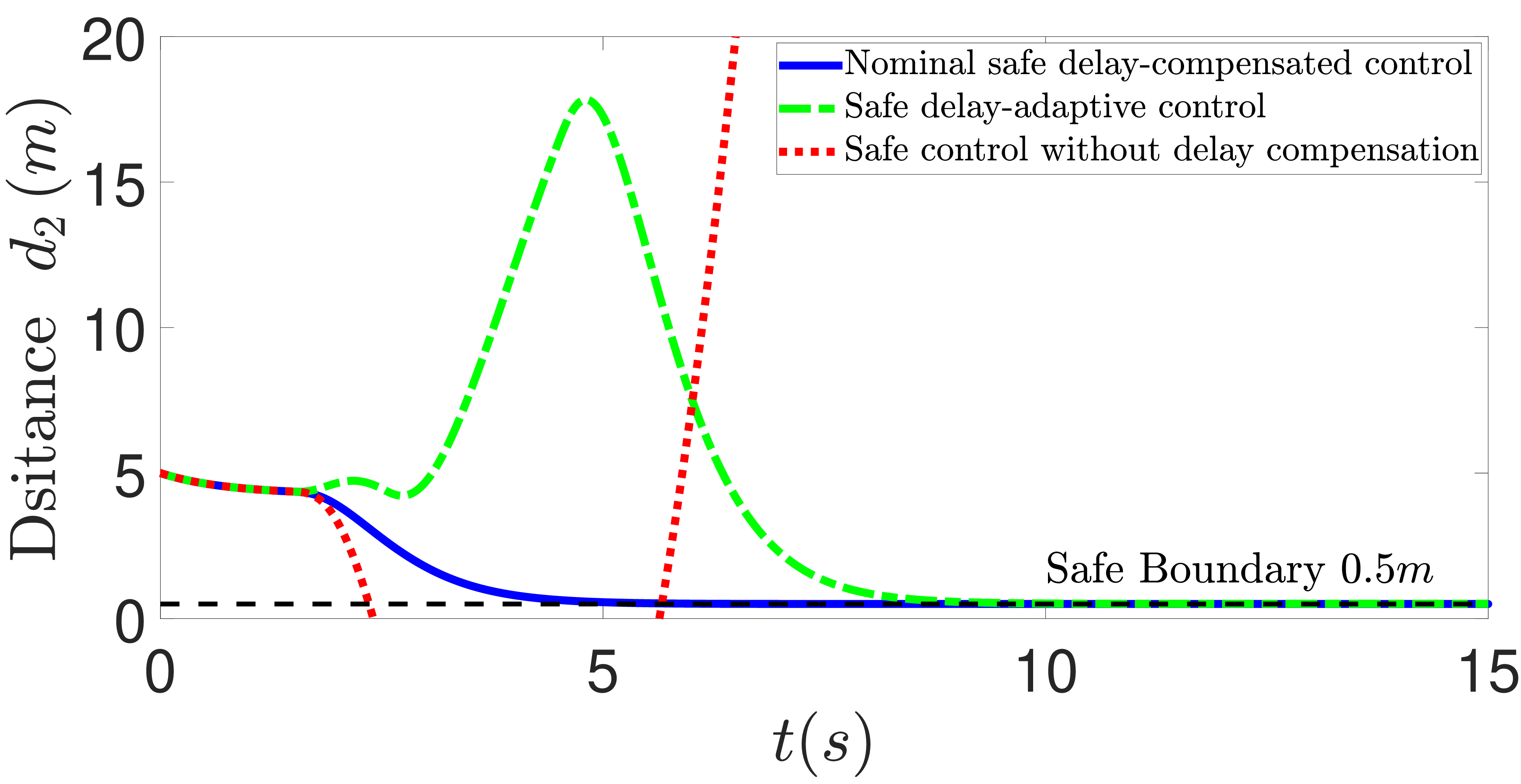}
	\caption{Results for the distance between vehicle $E_1$ and $E_2$, i.e., $d_2(t)=y_{21}(t)-y_{11}(t)$.}
	\label{fig:2}
	\subfloat[$v_1(t)$]{
		\includegraphics[width=0.48\linewidth, height=0.14\textheight]{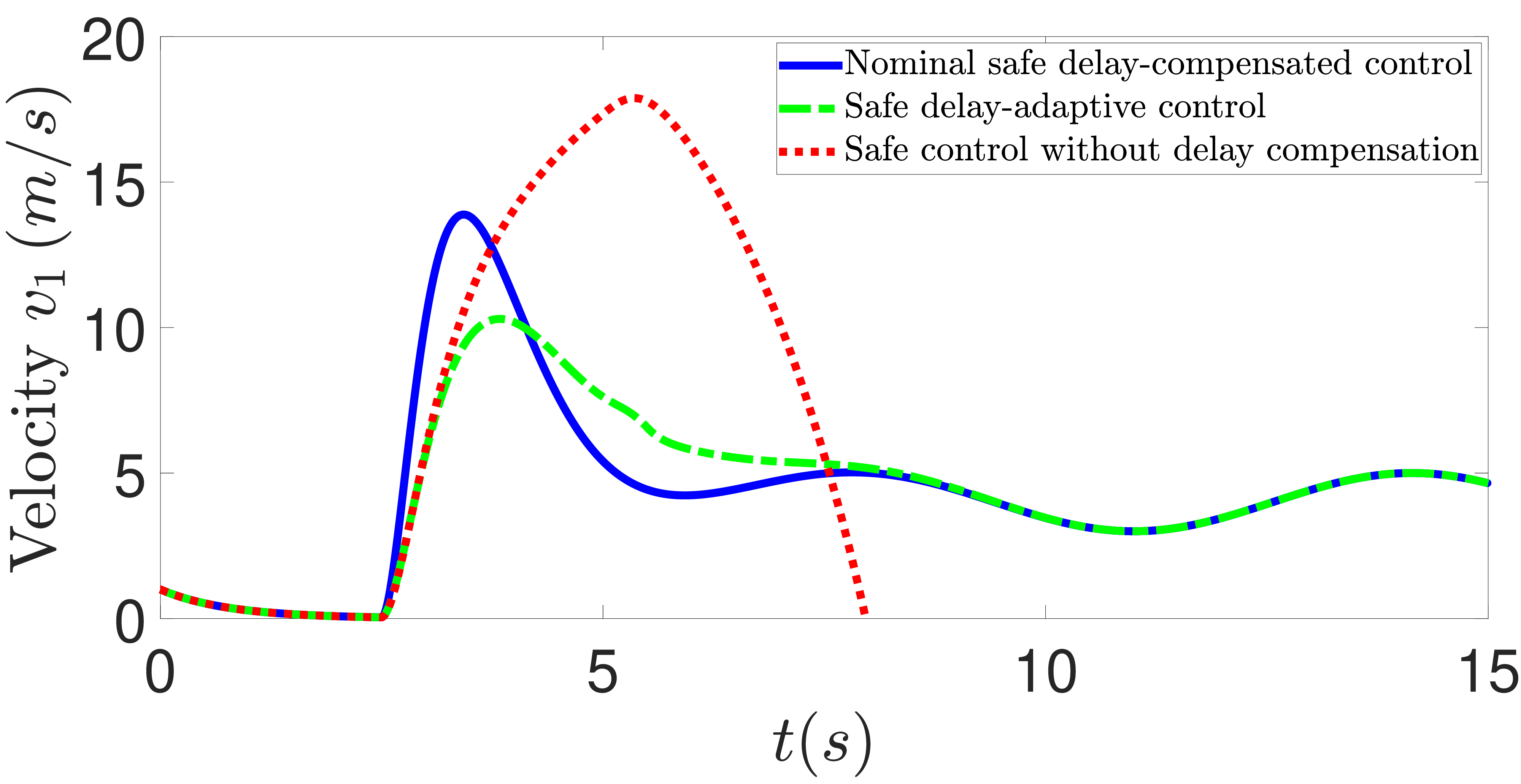}
		\label{fig:3.a}	}
	\subfloat[$v_2(t)$]{
		\includegraphics[width=0.48\linewidth, height=0.14\textheight]{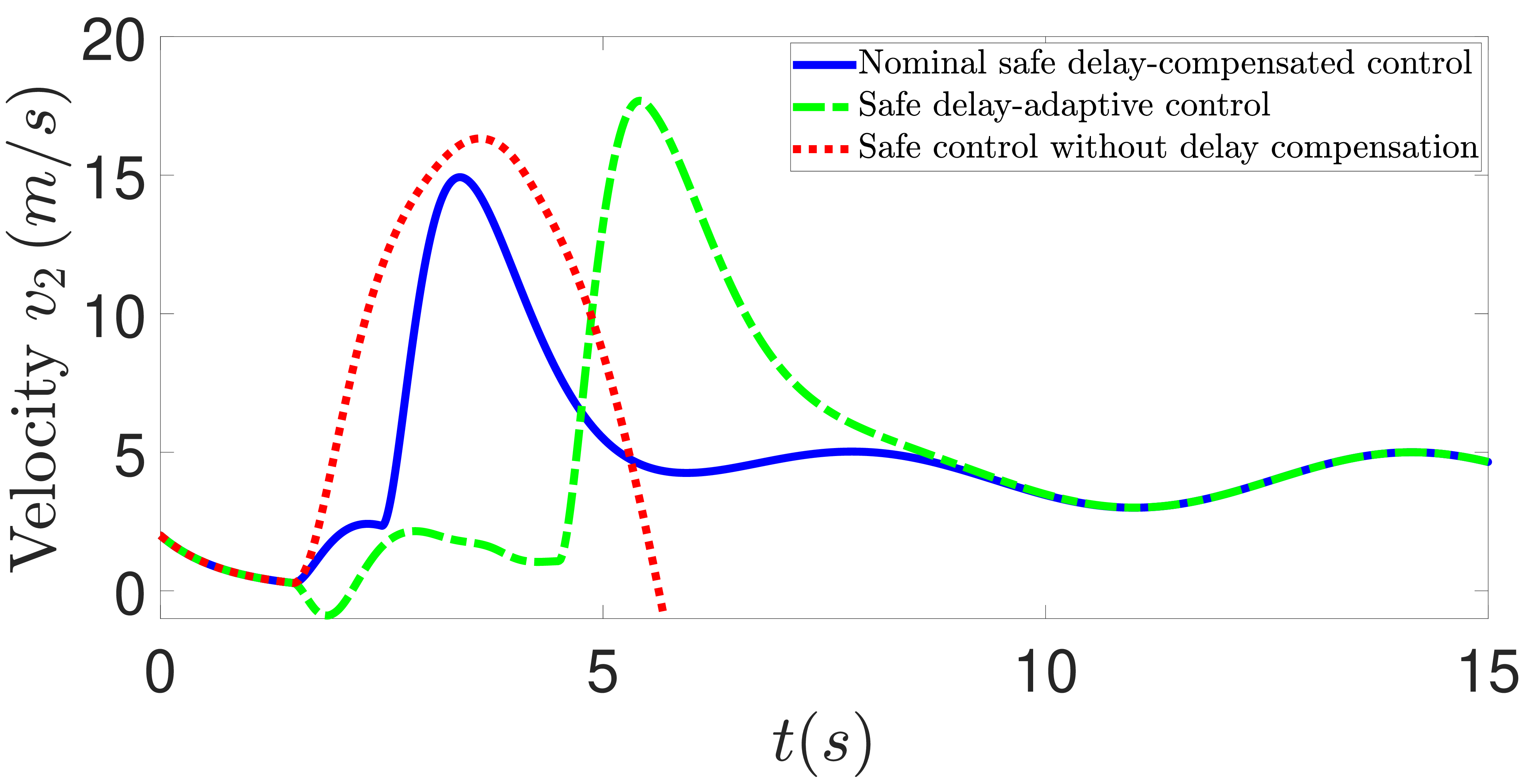}
    \label{fig:3.b}}
	\caption{Results for the velocities of followers $E_i$, i.e, $v_i=-y_{i2}$, $i=1,2$. }					\label{fig:3}
\end{figure}
\begin{figure}[!t]
		\centering	
		\subfloat[$F_1(t)$]{
		\includegraphics[width=0.48\linewidth, height=0.14\textheight]{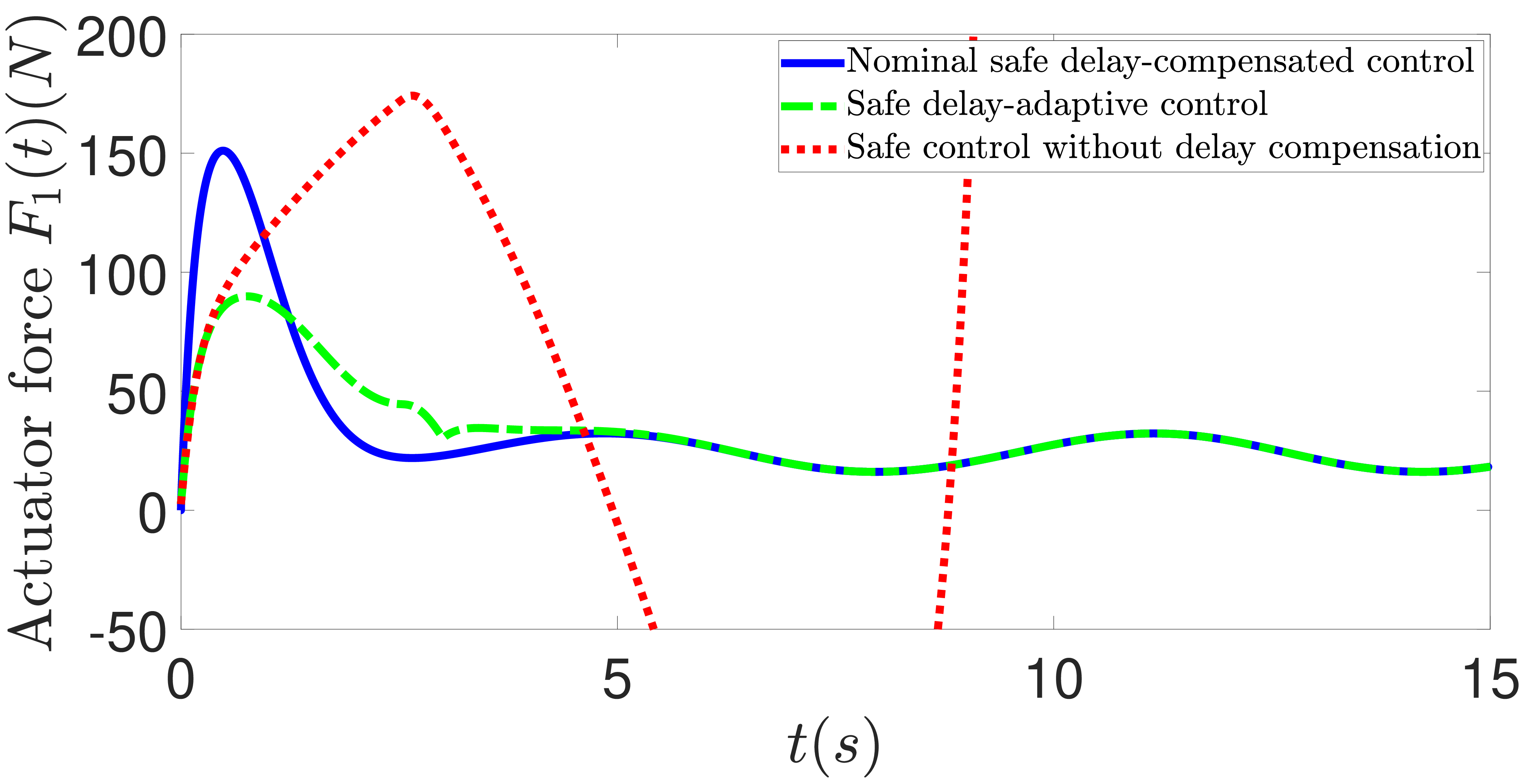}
		\label{fig:4.a}
	}
	\subfloat[$F_2(t)$]{
		\includegraphics[width=0.48\linewidth, height=0.14\textheight]{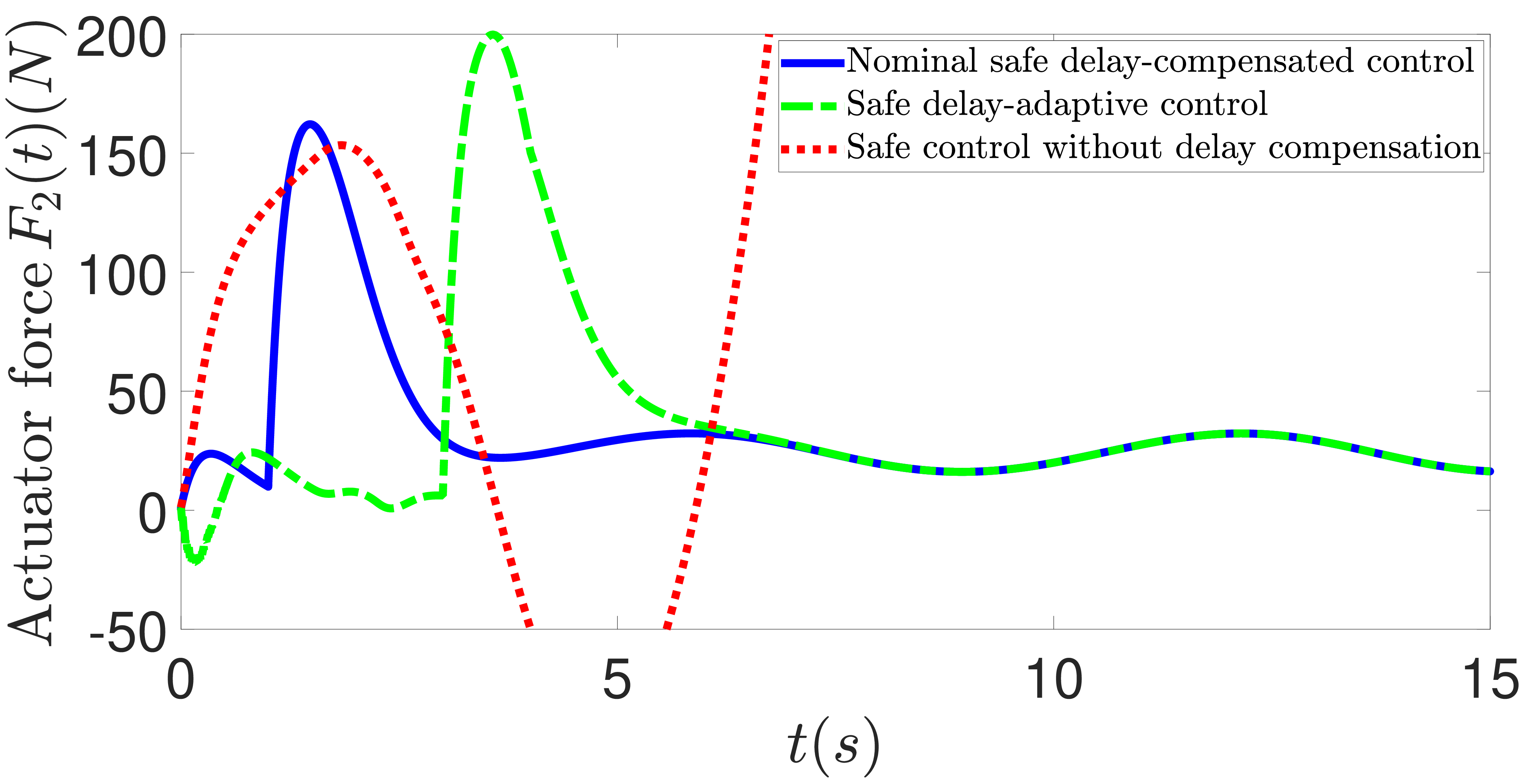}
		\label{fig:4.b}}\\
		\subfloat[$\mathcal{V}_1(t)$]{
		\includegraphics[width=0.48\linewidth, height=0.14\textheight]{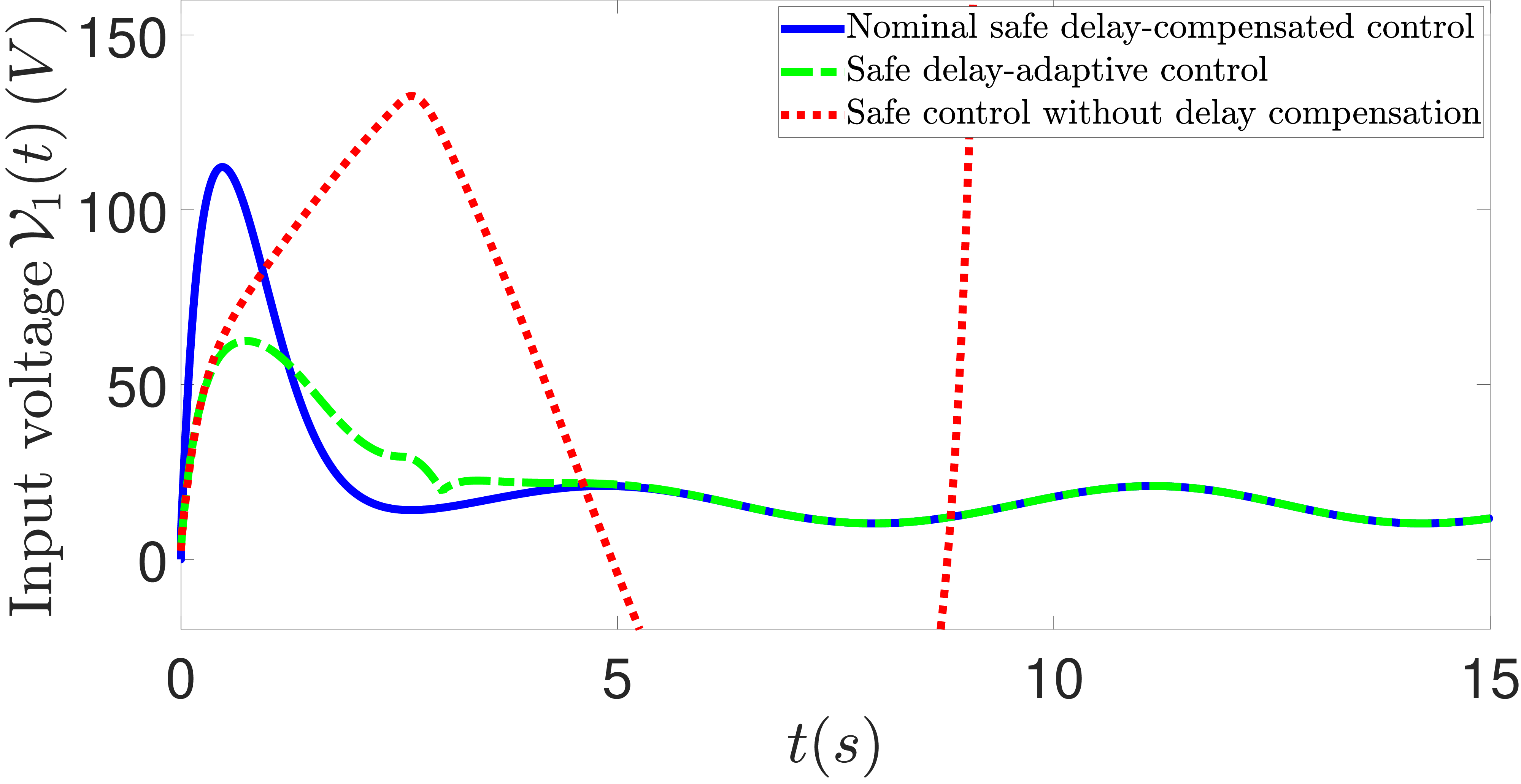}
		\label{fig:4.c}
	}
	\subfloat[$\mathcal{V}_2(t)$]{
		\includegraphics[width=0.48\linewidth, height=0.14\textheight]{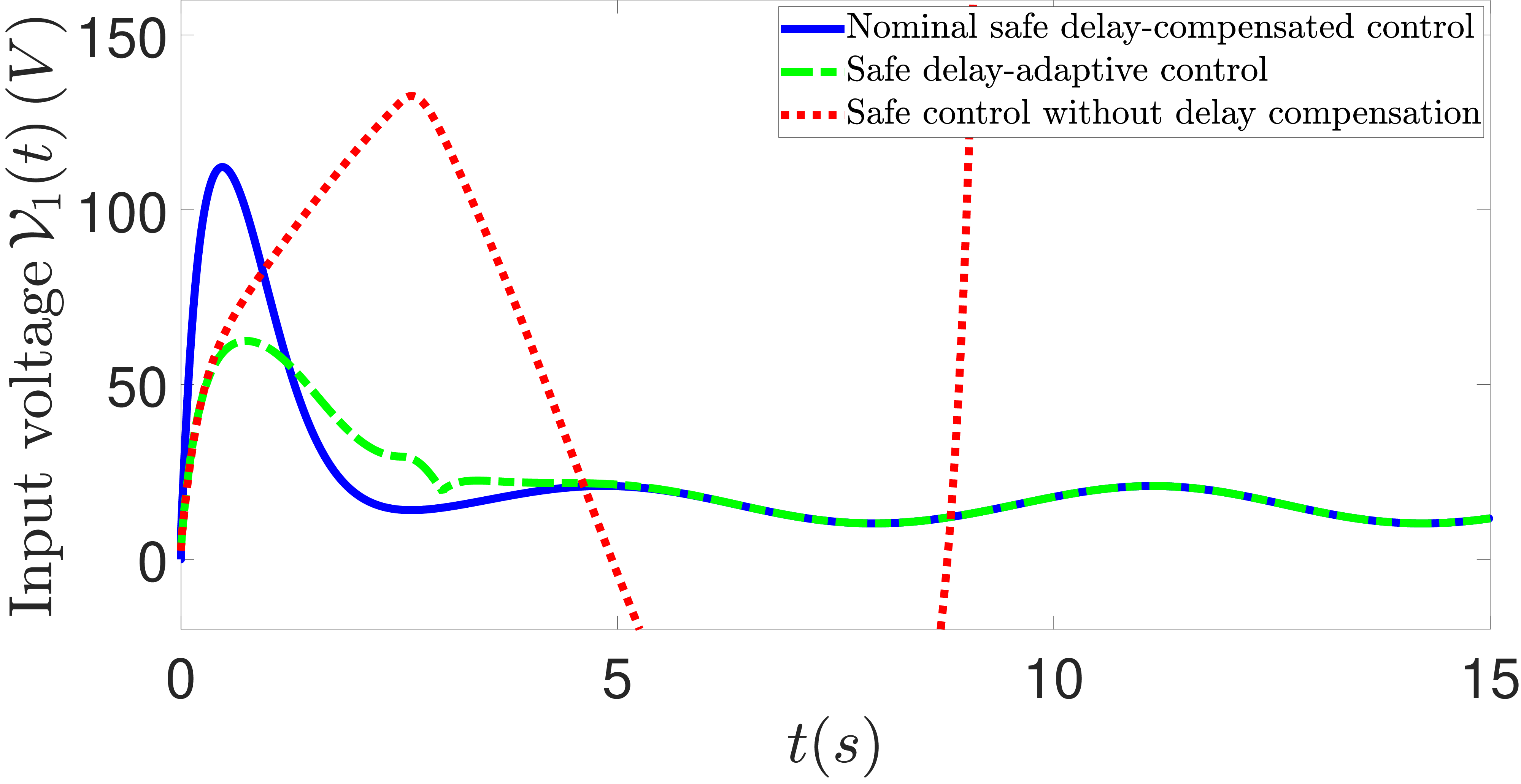}
		\label{fig:4.d}}
	\caption{Results for  the output force of the actuator  $F_i=x_{i1}(t)$ and the input voltage $\mathcal{V}_i(t)=\frac{rLU_i(t)}{k_t}$.}				\label{fig:4}
	\includegraphics[width=0.9\linewidth, height=0.23\textheight]{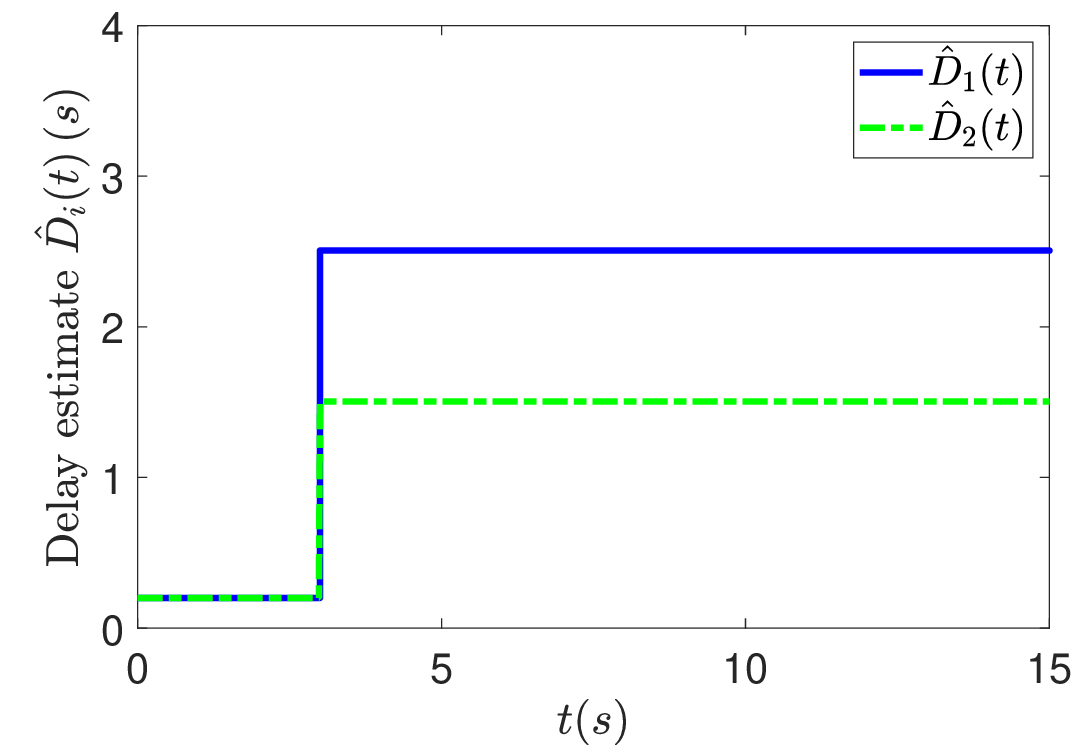}
	\caption{Estimates of the unknown delays $D_i,i=1,2$ under the initial estimates $\hat{D}_i(t_0)=0.2$.}		
	\label{fig:5}
\end{figure}

\subsection{Simulation results}
The simulation, including the implementation of the predictor \eqref{pred1} and identifier \eqref{iden}, is performed using the finite difference method with a time step $\diff t=0.001$ and a space step $\diff x=0.02$. As mentioned in Remark \ref{remarkimple},  the interval of $\mathcal{D}(i)$ is taken as $\diff _D=0.01$.
In addition to the nominal safe delay-compensated controller and the safe delay-adaptive controller, as a comparison, we also apply a safe controller without delay compensation, i.e., replacing the predictor states in \eqref{ut} with the states $Y(t)$.

 The simulation results are shown in Figs. \ref{fig:1}--\ref{fig:4}, where the blue line represents the results under nominal safe delay-compensated control, the green dot-dashed line denotes the results under safe delay-adaptive control, and the red dotted line shows the results under safe control without delay compensation. The results regarding the output states $y_{i1}(t),i=1,2$ of the $i$-th plant considered in this paper, i.e., the distances $d_i=y_{i1}+l_{i-1}$ between vehicles $E_i$ and $E_{i-1},i=1,2$  in practice,  are illustrated in Figs. \ref{fig:1}, \ref{fig:2}. We can see that the vehicle distances are convergent to the pre-set safe values $d_{o1}, d_{o2}$ respectively, and never exceed the safe boundary in the entire control process under the nominal and adaptive controllers. For the safe controller without delay compensation, the distances between vehicles undergo large oscillations, breaching the safety constraint and ultimately diverging due to the effects of the delay and nonlinearity. Compared with the nominal control, even though the results under the safe delay-adaptive controller exhibit greater conservatism with respect to safety in the process of delay identification, they have similar behavior ultimately, after the effective estimate of the unknown is obtained. Due to the uncertainty of the predicted information of vehicle $E_1$, which is required in the controller of $E_2$, the adaptive results of $E_2$ exhibit greater conservatism with respect to safety than $E_1$ in the process of delay identification, as shown in Figs. \ref{fig:1}, \ref{fig:2}. We also know from Figs. \ref{fig:3.a}, \ref{fig:3.b} that the velocities of follower vehicles $E_1, E_2$ converge to the target speed (i.e.,$v(t)=4+\sin t$) of the leader vehicle $E_0$ under adaptive and nominal control, while they diverge in the case without delay compensation. In Figs. \ref{fig:1}--\ref{fig:3}, the results under the three controllers are identical before the delay time $t=D_i$ because the vehicle behaviors only depend on the initial data on $t\in[0,D)$ and there is no the control action. The responses of the actuator output $F_i=x_{i1}(t), i=1,2,$ are shown in Figs. \ref{fig:4.a}, \ref{fig:4.b}, where they converge to constant values under adaptive and nominal control. The input voltages (i.e., control signals) of the two vehicles are depicted in Figs. \ref{fig:4.c}, \ref{fig:4.d} respectively. 
 Additionally, Fig. \ref{fig:5} presents the delay estimates from the delay identifier, where the blue line shows the estimate $\hat{D}_1(t)$ of the delay $D_1$ and the green dot-dashed line denotes the delay estimate $\hat{D}_2(t)$. Starting from the initial delay estimate $\hat{D}_i(t_0)=\underline{D}=0.2$, the successful identification of the unknown delay $D_i$ is achieved at the first triggering time $t_f=3$. Tiny differences exist between the delay estimates and the true values due to the errors in approximating integration as summation in the use of the finite difference method. 

\section{Conclusion and Future Work}\label{con}
In this paper, we design a safe delay-adaptive controller for a strict-feedback nonlinear system under a delayed nonlinear actuator, where the arbitrarily long delay $D$ between two nonlinear subsystems is unknown, based on the safe predictor-based backstepping transformation and a QP safety filter with BaLSI. Finally, we achieve exponential regulation of system states with a safety guarantee. The effectiveness of our design is verified in the application of safe vehicle platooning, ensuring vehicle string stability with a small gap and avoiding collisions at a relatively high speed in the presence of unknown delays.
In future work, we will consider external disturbances and measurement errors, which often occur in practical control design applications. We will also try to apply neural operators to improve the proposed controller's real-time efficacy in implementation by approximating the nonlinear ODEs as an open-loop flow map in the predictor \cite{neur1,neur2}.

\appendix
\begin{appendices}
\section{Proof of Proposition \ref{leminver1}} \label{dix1}
\subsection{Inverse transformation \eqref{inveryi}--\eqref{inverh}}\label{sec:inv1}
Considering $y_1(t)=z_1(t)+s(t)$ given by \eqref{zi} at $i=1$, we have $\psi_1(y_1(t))=\psi_1(z_1(t)+s(t)):=\bar{\psi}_1(z_1(t),s(t))$ which is continuously differentiable and satisfies $\bar{\psi}_1(0,0)=\psi_1(0)=0$ according to Assumption \ref{assum1}. Considering \eqref{equ1} in the original system and \eqref{obj1} in the target system, recalling \eqref{zi}, for $i=1$, we have
		\begin{align}
			y_2(t)&=\dot y_1(t)-\psi_1=\dot{z}_1(t)+s^{(1)}(t)-\bar{\psi}_1(z_1(t),s(t))\notag\\
			&=z_2(t)+\bar{h}_1({z}_1(t),{s}(t))+s^{(1)}(t),\label{inver y2}
		\end{align} 
		where $\bar{h}_1({z}_1(t),{s}(t))=-k_1z_1(t)-\bar{\psi}_1(z_1(t),s(t))$. From \eqref{inver y2} and $y_1(t)=z_1(t)+s(t)$, we have $\psi_2(\underline{y}_2(t)):=\bar{\psi}_2(\underline{z}_2(t),\underline{s}^{(1)}(t))$ by replacing $y_1,y_2$ with $z_1,z_2$ and $\underline{s}^{(1)}(t))$. It is obvious that $\bar{h}_1$,  $\bar{\psi}_2$ is continuously differentiable and $\bar{h}_1(0,0)=0$, $\bar{\psi}_2(0,0)=0$ from \eqref{inver y2}, Assumption \ref{assum1}, and $\bar{\psi}_1(0,0)=0$ shown above.
		Similarly, from \eqref{equ1}, \eqref{obj1} at $i=2$ and \eqref{inver y2}, one gets
		\begin{align}
			y_3(t)=&\dot y_2(t)-\psi_2=z_3(t)-k_2z_2(t)+s^{(2)}(t)-\bar{\psi}_2\notag\\
			&+\frac{\partial \bar{h}_{1}}{\partial z_1}\big(-k_1z_1+z_2\big)+\frac{\partial \bar{h}_{1}}{\partial s}s^{(1)}(t)\notag\\
			=&z_3(t)+\bar{h}_2(\underline{z}_2(t),\underline{s}^{(1)}(t))+s^{(2)}(t),\label{inver y3}
		\end{align}
		where $\bar{h}_2=-k_2z_2(t)-\bar{\psi}_2+\frac{\partial \bar{h}_{1}}{\partial z_1}\big(-k_1z_1+z_2\big)+\frac{\partial \bar{h}_{1}}{\partial s}s^{(1)}(t)$. 
		
We now prove the induction step: if  all the inverse transformations from $z_j(t)$ to $y_j(t)$, $j=2,\cdots,i$ for $i\le n-1$ are given as 
	\begin{equation}
		y_j(t)=z_j(t)+\bar{h}_{j-1}(\underline{z}_{j-1},\underline{s}^{(j-2)})+s^{(j-1)}(t),\label{inveryii}
	\end{equation}
	where $\bar{h}_{j}=-k_{j}z_{j}(t)-\bar\psi_{j} +\sum_{k=1}^{j-1}\big(\frac{\partial \bar{h}_{j-1}}{\partial z_k}\allowbreak\big(-k_kz_k+z_{k+1}(t)\big)+\frac{\partial \bar{h}_{j-1}}{\partial s^{(k-1)}}s^{(k)}(t)\big)$, $\bar{\psi}_{j}(\underline{z}_{j}(t),\underline{s}^{({j-1})}(t))=\psi_{j}(\underline{y}_{j}(t))$, $j=1,\cdots,i-1$ and for all $j$, $\bar{\psi}_{j}, \bar{h}_{j} $ are continuously differentiable function satisfying $\bar{h}_{j}(0,0)=0$, $\bar{\psi}_{j}(0,0)=0$, thus we have $y_{i+1}(t)=z_{i+1}(t)+\bar{h}_{i}+s^{(i)}(t)$  with continuously differentiable function $\bar h_i$, $\bar{\psi}_i$ satisfying $\bar h_i(0,0)=0, \bar{\psi}_i(0,0)=0$ as well. The proof of the induction step is given as follows. 
Substituting the induction hypothesis \eqref{inveryii} at $j=i$ into the \eqref{equ1} at $i+1$ in the original system, recalling \eqref{obj1} in the target system, one obtains 
$
		y_{i+1}(t)=\dot y_i(t)-\psi_i({\underline{y}_i})
		=z_{i+1}(t)+\bar{h}_i(\underline{z}_i(t),\underline{s}^{(i-1)}(t))+s^{(i)}(t)
$
	where $\bar{h}_i(\underline{z}_i(t),\underline{s}^{(i-1)}(t))=-k_iz_i-\bar\psi_i+\sum_{k=1}^{i-1}\big(\frac{\partial \bar{h}_{i-1}}{\partial z_k}\allowbreak\big(-k_kz_k+z_{k+1}\big)+\frac{\partial \bar{h}_{i-1}}{\partial s^{(k-1)}}s^{(k)}(t)\big)$,
	and $\bar{\psi}_{i}(\underline{z}_{i}(t),\underline{s}^{({i-1})}(t))=\psi_{i}(\underline{y}_{i}(t))$ by replacing $\underline{y}_{i}(t)$ in $\psi$ with $\underline{z}_{i}(t)$, $\underline{s}^{({i-1})}(t))$ using the relations from $z_j$ to $y_j$, $j=2,\cdots,i$, given by \eqref{inveryii}. The function $\bar\psi$ is continuously differentiable
 because all $\bar h_j, j=1,\cdots,i-1$ is continuously differentiable, and satisfies $\bar{\psi}_{i}(0,0)=\psi_{i}(0)=0$ because of \eqref{inveryii}, $\bar h_j(0,0)=0, j=1,\cdots,i-1$, and Assumption \ref{assum1}. It implies that the function $\bar h_{i}$ is continuously differentiable and $\bar h_i(0,0)=0$. The proof of the induction step is complete.
 
Starting from the base cases $y_1(t)=z_1(t)+s(t)$, \eqref{inver y2}, \eqref{inver y3}, and applying the induction step proved above, the inverse transformation \eqref{inveryi}--\eqref{inverh} is verified. 
\subsection{Inverse transformation \eqref{inverker}}
Recalling \eqref{equ2}, \eqref{uxt}, we have
\begin{align}
    \dot y_n(t+Dx)=\psi_n(\underline{y}_n(t+Dx))+u(x,t).\label{eq:ynin}
\end{align}
Taking the time derivative of \eqref{inveryi} at $i=n$, replacing the current states $y_n,z_i$ by the predictor states $p_n,\delta_i$, applying $\dot p_n(x,t)=\psi_n(\underline{p}_n(x,t))+u(x,t)$ obtained from \eqref{eq:ynin}, we then have
	\begin{align}	&\psi_n(\underline{p}_n(x,t))+u(x,t)=\frac{\partial {\delta_n}(x,t)}{\partial t}+s^{(n)}(t+Dx)\notag\\
	&+\sum_{k=1}^{n-1}\big[\frac{\partial \bar{h}_{n-1}(\underline{\delta}_{n-1}(x,t),\underline{s}^{(n-2)}(t))}{\partial \delta_k(x,t)}\big(-k_k\delta_k(x,t)\notag\\
	&+\delta_{k+1}(x,t)\big)+\frac{\partial \bar{h}_{n-1}}{\partial s^{(k-1)}}s^{(k)}(t+Dx)\big].\label{eq:pnxt}
	\end{align}
	Recalling $\psi_n(\underline{p}_n(x,t))=\bar{\psi}_{n}(\underline{\delta}_{n}(x,t),\underline{s}^{({n-1})}(t+Dx))$, plugging $\frac{\partial {\delta_n}(x,t)}{\partial t}=-k_n\delta_n(x,t)+w(x,t)$ obtained  from \eqref{obj2} into \eqref{eq:pnxt}, one gets
		\begin{align}
				&u(x,t)=w(x,t)+\Big(-k_n\delta_n(x,t)-\bar{\psi}_{n}\notag\\
		&+\sum_{k=1}^{n-1}\Big[\frac{\partial \bar{h}_{n-1}}{\partial \delta_k(x,t)}\big(-k_k\delta_k(x,t)+\delta_{k+1}(x,t)\big)\notag\\
		&+\frac{\partial \bar{h}_{n-1}}{\partial s^{(k-1)}}s^{(k)}(t+Dx)\Big]\Big)+s^{(n)}(t+Dx). \label{inveruxt}
		\end{align} 
Recalling the definition of $\bar{h}_n$ in \eqref{inverh}, the inverse of the transformation \eqref{ker} is obtained as \eqref{inverker}.

\subsection{Inverse transformation \eqref{inver ri}--\eqref{invertau}}	
According to \eqref{inverker}, \eqref{eq:ux1}, \eqref{obj4}, we have $x_1=\frac{1}{b}r_1+\frac{1}{b}\bar h_n(\delta(1,t),\underline{s}^{({n-1})}(t+D))+\frac{1}{b}s^{(n)}(t+D)$, i.e., \eqref{inver ri} at $j=1$ with \eqref{inverdelta}. Through the recursive process similar to Sec. \ref{sec:inv1}, recalling \eqref{obj5}, \eqref{obj6} and \eqref{equ3}, \eqref{equ4},  as well as Assumption \ref{assum1}, the inverse of transformation \eqref{ri}--\eqref{tau} is obtained as \eqref{inver ri}--\eqref{invertau}, with continuously differentiable $\bar{\varphi}_j({\underline{r}_j(t)},\underline{\Delta}^{(i-1)}(t))=\varphi_j({\underline{x}_j})$ satisfying $\bar{\varphi}_j(0,0)=0$.

\section{The proof of Lemma \ref{lem:targetstability}} \label{dixlem1}
We construct the following Lyapunov function for the target system \eqref{obj1}--\eqref{obj6}, 
	\begin{align}
	V(t)=&\frac{1}{2}\sum_{i=1}^{n}z_i(t)^2+\frac{\rho}{2}\sum_{i=1}^{m}r_i(t)^2\notag\\
            &+\frac{1}{2}\sum_{i=0}^{m}\int^1_0a_ie^{ x} w^{(i)}_x(x,t)^2 \diff x,\label{v}
	\end{align}
	where  $\rho,a_1,\cdots,a_m$ are  positive analysis parameters that will be determined later. 
	According to \eqref{ome1},	we have
	\begin{equation}
		\theta _1 \Omega(t)\leq V(t)  \leq \theta _2 \Omega(t),\label{ome2}
	\end{equation}
	for some positive constants $\theta _1, \theta _2$.
	Taking the time derivative of $V(t)$ in \eqref{v}, one obtains 
	\begin{align}
			&\dot V(t)=-\sum_{i=1}^{n}k_iz_i(t)^2+\sum_{i=1}^{n-1}z_i(t)z_{i+1}(t)+w(0,t)z_n(t)\notag\\
			&-\rho\sum_{i=1}^{m}c_ir_i(t)^2+\rho\sum_{i=1}^{m-1}r_i(t)r_{i+1}(t)+\frac{a_0e}{2D} w(1,t)^2\notag\\
			&-\sum_{i=0}^{m}\frac{a_i}{2D}w_x^{(i)}(0,t)^2+\sum_{i=1}^{m}\frac{a_ie}{2D} w_x^{(i)}(1,t)^2	\notag\\
			&-\sum_{i=0}^{m}\frac{a_i}{2D} \int_{0}^{1} e^{ x} w_x^{(i)}(x,t)^2 \diff x,\label{v1}
	\end{align}
	where integration by parts has been used. Recalling \eqref{obj4}--\eqref{obj6}, we conclude that there exist positive constants $b_1,\cdots,b_m,$ determined by delay time $D$, design parameters $c_1,\cdots,c_m,$ and analysis parameters $a_1,\cdots,a_m,$ such that
	\begin{align}
			\sum_{i=1}^{m}\frac{a_ie}{D} w_x^{(i)}(1,t)^2=\sum_{i=1}^{m}a_ieD^{i-1} r_1^{(i)}(t)^2\leq\sum_{i=1}^{m}b_ir_i(t)^2.\label{bi}
	\end{align}
	Thus applying Young's inequality and inserting \eqref{bi} into \eqref{v1} yield that               
	\begin{align}
			\dot V(t)&\le-(k_1-\frac{1}{2})z_1(t)^2-\sum_{i=2}^{n}(k_i-1)z_i(t)^2\notag\\
			&-\left[\rho(c_1-\frac{1}{2})-\frac{a_0e}{2D}-\frac{1}{2}b_1\right]r_1(t)^2\notag\\
			&-\sum_{i=2}^{m-1}[\rho(c_i-1)-\frac{1}{2}b_i]r_i(t)^2\notag\\
			&-\left[\rho(c_m-\frac{1}{2})-\frac{1}{2}b_m\right]r_m(t)^2-(\frac{a_0}{2D}-\frac{1}{2})w(0,t)^2\notag\\
			&-\sum_{i=0}^{m}\frac{a_i}{2D} \int_{0}^{1} e^{ x} w_x^{(i)}(x,t)^2 \diff x-\sum_{i=1}^{m}\frac{a_i}{2D}w_x^{(i)}(0,t)^2.\label{lya}
	\end{align}
	Under the conditions of the design parameters $c_1,\cdots, c_m,\allowbreak k_1,\cdots,k_n$ in \eqref{ki}--\eqref{c}, and choosing the analysis parameters $a_0,\cdots,a_m,\rho$ as 
    \begin{align}
		a_0&\geq D,\quad a_i>0,\,i=1,\cdots,m\label{ai}\\
		\rho&>\max\left\{ \frac{a_0e}{3D}+\frac{1}{3}b_1,b_m,\frac{1}{2}b_j \right\}+1, j=2,\cdots,m-1 \label{rho}
    \end{align}
	we get
	\begin{align}
			\dot V(t)\leq&-\sum_{i=1}^{n}z_i(t)^2-\sum_{i=1}^{m}r_i(t)^2\notag\\	
			&-\sum_{i=0}^{m} \frac{a_i}{2D}\int_{0}^{1} e^{ x} w_x^{(i)}(x,t)^2 \diff x\leq-\varrho V(t),\label{eq:Ly1}
	\end{align}
	where $\varrho=\frac{1}{\theta_2}\min\left\{1, \frac{a_i}{2D} \right\}>0, i=0,\cdots,m.$	
	Recalling \eqref{ome2}, we thus have \eqref{ome3}, where $\Upsilon_{\Omega}=\frac{\theta_2}{\theta_1}$ and $\sigma_{\Omega}=\varrho$. The lemma is obtained.
 
\section{The proof of Lemma \ref{lemma:deltarget}} \label{dixlem2}
 Applying Cauchy-Schwarz inequality for \eqref{inverpred},  it is straightforwardly obtained that $|\delta(x,t)|^2$, $\forall x\in[0,1]$, are exponentially convergent to zero from the exponential convergence to zero of $|Z(t)|^2$, $\Vert w(\cdot,t)\Vert^2$ in Lemma \ref{lem:targetstability}. Moreover, taking $i$ order time derivatives  of \eqref{inverpred} at $x=1$, one obtains 
	\begin{align}
		&\delta_t^{(i)}(1,t)=e^{DA}Z^{(i)}(t)+D\int_{0}^{1}e^{DA(1-y)}Bw_t^{(i)}(y,t)\diff y\notag\\
		&=e^{DA}\big(A^iZ(t)+\sum_{j=0}^{i-1}A^{i-1-j}B{w_t^{(j)}}(0,t)\big)+\frac{1}{D^{(i-1)}}\notag\\
		&\times\Big[\sum_{j=0}^{i-1}(DA)^j\big(Bw_x^{(i-1-j)}(1,t)-e^{DA}Bw_x^{(i-1-j)}(0,t)\big)\notag\\
		&+\int_{0}^{1}(DA)^ie^{DA(1-y)}Bw(y,t)\diff y\Big],\label{deltai}
	\end{align}
where integration by parts and \eqref{obj3} have been used. According to \eqref{obj3}--\eqref{obj6}, applying Cauchy-Schwarz inequality, we also have
$
w_x^{(i)}(1,t)^2+
    w_x^{(i)}(0,t)^2 \le \Upsilon_w\left(\|w_x^{(i+1)}(\cdot,t)\|^2+\underline {r}_{i+1}(t)^2\right),~i=1,\cdots,m-1
$
for some positive $\Upsilon_w$. Then applying 
\eqref{deltai}, and recalling the exponential convergence to zero of $\sum_{i=0}^{m}\Vert w_x^{(i)}(\cdot,t) \Vert^2$, $|Z(t)|^2$, $|R(t)|^2$ in Lemma \ref{lem:targetstability}, and that of $|\delta(x,t)|^2,\forall x\in[0,1]$ obtained above, we have that $|\delta_t^{(i)}(1,t)|^2,i=0,\cdots,m$ are exponentially convergent to zero. The lemma is then obtained.
\end{appendices}

\bibliographystyle{bibsty}                     
\bibliography{Safe_Delay-Adaptive_Control}           

\end{document}